\documentclass[11pt]{article}
\pdfoutput=1
\usepackage[utf8]{inputenc}
\usepackage[T1]{fontenc}
\usepackage{fullpage}
\usepackage[english]{babel}
\usepackage{verbatim}
\usepackage{amsmath,amsthm,amssymb,amsfonts,hyperref,physics,algorithm}
\usepackage{graphicx, xcolor}
\usepackage{subcaption}
\usepackage[noend]{algpseudocode}

\usepackage[capitalize]{cleveref}
\newcommand{\cE}{{\mathcal{E}}  }
\newcommand{\dt}{\delta t  }
\newcommand{\bl}[1]{ {\color{black} #1} }
\newcommand{\red}[1]{ {\color{black} #1} }
\newcommand{\teal}[1]{ {\color{black} #1} }

\newtheorem{theorem}{Theorem}

\newtheorem*{theorem*}{Theorem}
\newtheorem{lemma}[theorem]{Lemma}

\makeatletter
\def\BState{\State\hskip-\ALG@thistlm}
\makeatother

\title{
Non-Markovian Noise Mitigation: \\
Practical Implementation, Error Analysis, and the Role of Environment Spectral Properties
}
\author{
	Ke Wang\thanks{Department of Mathematics, The Pennsylvania State University \texttt{wangke.math@psu.edu} }
	\and
	Xiantao Li\thanks{Department of Mathematics, The Pennsylvania State University \texttt{xiantao.li@psu.edu}}
}
\date{\today\vspace{-5mm}}
\begin{document}
\maketitle
\begin{abstract}
Quantum error mitigation(QEM),  an error suppression strategy without the need for additional ancilla qubits for noisy intermediate-scale quantum~(NISQ) devices, presents a promising avenue for realizing quantum speedups of quantum computing algorithms on current quantum devices. However, prior investigations have predominantly been focused on Markovian noise, which only occurs when the separation between the system and environment is sufficiently large.  In this paper, we propose a non-Markovian Noise Mitigation(NMNM) method by extending the probabilistic error cancellation (PEC) method in the QEM framework to treat non-Markovian noise. We present the derivation of a time-local quantum master equation where the incoherent coefficients are directly obtained from bath correlation functions(BCFs), key properties of a non-Markovian environment that will make the error mitigation algorithms environment-aware. We further establish a direct connection between the overall approximation error and sampling overhead of QEM and the spectral property of the environment. Numerical simulations performed on a spin-boson model further validate the efficacy of our approach. 
\end{abstract}
\section{Introduction}
Quantum algorithms have established theoretical advantage over classical methods in tasks such as integer factorization~\cite{shor1999polynomial} and quantum simulations~\cite{low2019hamiltonian,lloyd1996universal,berry2015simulating,low2017optimal}, yet the experimental realization of this advantage remains elusive. Current quantum platforms face two major obstacles: the limited scalability of existing hardware and the inevitable effects of noise in quantum circuits. Quantum error correction (QEC)~\cite{nielsen2010quantum} suppresses the noise with qubits overhead, which requires extra scalability of the quantum computer. In contrast, quantum error mitigation (QEM)~\cite{temme2017error,endo2018practical,cai2023quantum}, tailored for noisy intermediate-scale quantum~(NISQ) devices~\cite{preskill2018quantum}, reduces the effective noise through a sampling overhead rather than an increase in qubit resources. 

Continuous QEM is a scheme that has been observed to be compatible with digital quantum computers ~(DQC) and certain regimes of analog quantum computers ~(AQC)~\cite{sun2021mitigating}. In DQC, gate-based quantum circuits are implemented by using digital pulses, where the noises are simplified to a quantum channel before or after the ideal gate~\cite{garcia2024mitigating}. Both DQC and AQC inherently involve continuous quantum state evolution, allowing master-equation-based treatments of noise to closely mimic experimental procedures~\cite{cai2023quantum}.    
In this framework, the change of the quantum state on a noiseless circuit follows the time evolution according to the Liouville von Neumann equation,
\begin{equation}\label{eq: ideal-rho}
    \partial_t\rho_I(t) = -i\left[H_S,\rho_I(t)\right]. 
\end{equation}
This ideal operation, without the interference from the environment noise,  is denoted by a unitary channel $\mathcal{E}_I$,  such that $\rho_I(t+\delta t) = \mathcal{E}_I(t+\delta t, t)\rho_I(t)$ for any time interval $\delta t$. 

On the other hand, in the presence of Markovian noise, the dynamics of the state (denoted by $\rho_N$) follows the Lindblad dynamics,
\begin{equation}\label{eq: lindblad}
    \partial_t\rho_N(t) = -i\left[H_S,\rho_N(t)\right] + \mathcal{L}_N\rho_N(t).
\end{equation}
Similarly, the noisy operation induces a quantum channel $\mathcal{E}_N$ such that $\rho_N(t+\delta t) = \mathcal{E}_N(t+\delta t, t)\rho_N(t)$. A direct derivation of the Lindblad equation \cite{carmichael2013statistical} shows that  $\mathcal{L}_N$ could include both the {coherent} and {incoherent noise}.

In order to mitigate the noise introduced by the dissipative operator $\mathcal{L}_N$ in \cref{eq: lindblad}, the recovery operator $\mathcal{E}_Q(t,t+\delta t)$ can be introduced, with the idea to reverse the influence from the noise, i.e., $\mathcal{E}_Q(t,t+\delta t)\mathcal{E}_N(t+\delta t, t)\approx \mathcal{E}_I(t+\delta t, t)$ \cite{temme2017error}. The design of  $\mathcal{E}_Q(t,t+\delta t)$  often involves propagating the noise operator  $\mathcal{L}_N$ operator backward in time, and we reversed $t$ and $t+\delta t$ in $\cE_Q$ to reflect this perspective. 
Unfortunately, the recovery operation is generally a non-physical operation, in that it is no longer a quantum channel.  Temme et al. \cite{temme2017error} proposed the  \emph{probabilistic error cancellation} (PEC) method, where the recovery operator is projected to a physical basis, $\{\mathcal{B}_\ell\}_\ell$, which can be directly implemented on quantum devices.    Specifically,  $\mathcal{E}_Q(t,t+\delta t) = \gamma(t,\delta t)\sum_\ell \alpha_\ell(t,\delta t)p_\ell(t,\delta t)\mathcal{B}_\ell$ with probability $p_\ell$, the sign $\alpha_\ell = \pm 1$ and the normalization constant $\gamma(t,\delta t) $. Thus, the expectation of an observable $O$ on the state $\rho_I(t+\delta t)$ after one step $\delta t$ can be approximated as follows.
\begin{equation}\label{eq:introPEC}
    \tr\bigl(O\rho_I(t+\delta t)\bigr) = \gamma(t,\delta t) \sum_\ell\alpha_\ell p_\ell \tr\left(O\mathcal{B}_\ell\mathcal{E}_N(t)\rho_N(t)\right).
\end{equation}
In practice, this scheme is repeated for each layer that corresponds to a time step in continuous QEM and the ensemble can be implemented via a direct Monte Carlo sampling with respect to the joint density of $\{p_\ell\}$.  This PEC algorithm has been shown with improved gate fidelity in {superconducting qubits}~\cite{song2019quantum} and trapped ion~\cite{zhang2020error} systems,  and applied to various quantum algorithms, including variational quantum eigensolvers~\cite{strikis2021learning,jose2022error} and dynamic simulation of Ising spin chain~\cite{van2023probabilistic},  Fermi-Hubbard model~\cite{chen2023error} and Lindblad simulations~\cite{guimaraes2023noise}.

While quantum error mitigation (QEM) methods for Markovian noise are relatively well-established, the dynamics of many open quantum systems fall into the non-Markovian regime,  \teal{which arises when the coupling between the system and environment is  large. The corresponding strategies for general environment noise in this regime remain largely underexplored.} In contrast to Markovian systems, which can be described by the Lindblad form, non-Markovian open quantum systems often lack a universal master equation. Such observations have motivated many approaches to develop appropriate mathematical descriptions \cite{tanimura2020numerically,breuer2004time,suess_hierarchy_2014,li2021markovian,mascherpa2017open}. More important to QEM is the fact that non-Markovian dynamics often exhibit memory effects that complicate noise manipulation. A direct consequence of these memory effects is the breakdown of zero-noise extrapolation (ZNE) \cite{temme2017error}: The presence of memory prevents one from rescaling the system Hamiltonian to provide a noise-scaling factor, a technique used in \cite{temme2017error} to achieve amplification of the noise. As a result, commonly used gate-folding techniques cannot be directly applied in the presence of non-Markovian noise.

Meanwhile, probabilistic error cancellation (PEC) still remains a viable option for non-Markovian noise. A suitable description of non-Markovian noise and a proper choice of the basis operations $\mathcal{B}_\ell$ remain two important challenges to this problem. 
Hakoshima et al.\ \cite{hakoshima2021relationship} proposed a PEC  strategy with time-dependent $\mathcal{B}_\ell(t)$ that leverages the canonical form of the quantum master equation \cite{hall2014canonical}, derived under an invertibility assumption. One noteworthy observation in their work is that the incoherent parameter can become negative in the non-Markovian regime, which, within the PEC framework, {does not increase } the sampling overhead. However, unlike many studies of non-Markovian dynamics, where bath properties play a critical role, the approach \cite{hakoshima2021relationship} assumed knowledge of the incoherent coefficients without explicitly connecting them to bath properties. Consequently, the influence of bath properties on the resource overhead remains unclear.  Liu et al.\cite{liu2024non} developed a non-Markovian PEC method, by leveraging the Choi channel representation of a non-Markovian process. Further, by deriving a $\chi-$matrix representation,  Markovian PEC methods with standard operation basis \cite{endo2018practical} can be applied directly. 
Ahn and co-workers \cite{ahn2023non,ahn2024non} investigated the PEC approach with Dirac Gamma matrices for two-qubit gate operations. Their approach, however, is limited to a specific non-Markovian noise model (Caldeira–Leggett) and the analysis only considered single-step error mitigation, without identifying the sampling complexity in general.

In this paper, we propose a non-Markovian Noise Mitigation (NMNM) approach specifically designed for non-Markovian noise and our algorithm directly incorporates the environment properties. In particular:
\begin{itemize}
    \item \textbf{Time-dependent noise operator}: We present a straightforward derivation of a time-local quantum master equation with a time-dependent noise operator. This noise operator naturally separates into coherent and incoherent terms. Notably, the incoherent term can be written in the GKS form \cite{gorini1976completely},  albeit with possible negative incoherent coefficients, which are directly linked to the bath’s spectral properties.

    \item \textbf{PEC steps with a superoperator basis}: Unlike the approach in \cite{hakoshima2021relationship}, we perform the PEC steps by projecting the recovery operator onto a \emph{time-independent} superoperator basis \(\mathcal{B}_\ell\), as commonly employed in standard PEC procedures, thereby simplifying the implementation of the algorithm.

    \item \textbf{Resource overhead analysis}: We provide rigorous estimates of the approximation error, together with the resource overhead, expressed directly in terms of the bath correlation function and the spectral density. In particular, we introduce the concept of an \emph{effective spectral parameter} for the environment, $\teal{G_{\mathrm{env}}}$,  which comes from the analytic properties of the bath correlation function. We prove that the norm of this effective parameter provides a tight upper bound on both the approximation error and the overall sampling overhead. Therefore, our results reveal how certain spectral properties of the environment critically influence the effectiveness of QEM.

    \item \textbf{Noise simulation and validation}: We describe how to simulate non-Markovian noise on classical devices, enabling preliminary testing of QEM algorithms before deploying them on quantum hardware. Numerical experiments on the spin-boson model are presented to demonstrate the effectiveness of the proposed QEM method.
\end{itemize}

We summarize our theoretical results as follows.
\begin{theorem*}[Informal version of \cref{thm:bias,thm:gamma,thm:complexity}]
For a quantum circuit coupled to an environment with bath correlation functions~\cite{meier1999non, ritschel2014analytic} $C_{j,k}(t) = \sum_{\mu}g^*_{j,\mu}g_{k,\mu}e^{i\omega_\mu t}, t>0, \omega_\mu\in\mathbb{C}$,
    the Non-Markovian Noise Mitigation(NMNM) method with $N_r$ number of samples and discrete time interval $\delta t$ suffices to approximate the ideal expectation $\tr(O\rho_I(T))$  within an additive error $\epsilon$, provided that
   \begin{equation}
       \delta t = O\left( \frac{\epsilon}{\lambda^2(T+1/4)\teal{G_{\mathrm{env}}}}\right),\quad N_r  = \Omega\left(\frac{\exp(\lambda^2 T \teal{G_{\mathrm{env}}})}{\epsilon^2}\right),
   \end{equation}
   where the bound $\teal{G_{\mathrm{env}}}$ are  related to the spectral properties of the environment as follows
   \begin{equation}
       \teal{G_{\mathrm{env}}} = 2\sum_\mu(\sum_{j}|g_{j,\mu}|)^2\frac{1}{\Im(\omega_\mu)}.
   \end{equation}
\end{theorem*}

The exponential dependence of the sampling complexity meets the lower bound \cite{quek2024exponentially, takagi2023universal,takagi2022fundamental} in general. On the other hand, the exponent is also proportional to the environment parameter $\teal{G_{\mathrm{env}}}$, which we will refer to as an effective spectral parameter, and it provides an important guideline for implementing the QEM algorithm. In particular, $\teal{G_{\mathrm{env}}}$ depends on the locations of the poles of the spectral bath density.  

The organization of the remaining part of this paper is as follows. In \cref{Preliminary}, we review the idea of PEC. In \cref{NMEM}, the main method for NMNM is outlined. As our main theoretical contribution, the relation between the error, including the approximation and statistical error, and the spectral properties of the environment are discussed in \cref{AnalysisCost}. The performance of our QEM method is illustrated by numerical experiments based on some spin-boson models in \cref{NumericalExperiment}.

\section{A Brief Introduction to Probabilistic Error Cancellation(PEC) }\label{Preliminary}
In this section, we briefly review the setup of Markovian PEC. 
We discretize the total evolution time $T$ into $M$ steps so that $T = M\delta t$ with size $\delta t$. The ideal and noisy evolution of the quantum states from $t$ to $t+\delta t$ are denoted by $\mathcal{E}_I(t+\delta t, t)$ and $\mathcal{E}_N(t +\delta t, t)$ following \cref{eq: ideal-rho} and \cref{eq: lindblad}, respectively. Namely
\begin{equation} 
    \rho_I(t+\delta t) = \mathcal{E}_I(t+\delta t, t)\rho_I(t) , \quad \rho_N(t+\delta t)= \mathcal{E}_N(t+\delta t, t)\rho_N(t). 
\end{equation}

We aim to find a recovery operator $\mathcal{E}_Q(t, t+\delta t)$ to mitigate the error, i.e., 
\[
\mathcal{E}_Q(t, t+\delta t) \mathcal{E}_N(t+\delta t, t)  \approx \mathcal{E}_I(t+\delta t, t).
\]

For Markovian noise, $\mathcal{E}_Q \approx I - \delta t \mathcal{L}_D$, with $\mathcal{L}_D$ being a Lindblad operator \cite{temme2017error}. Although it is not a physical operator, it can be expressed as a linear combination of completely positive (CP) operator basis $\{\mathcal{B}_\ell:  \mathcal{B}_\ell\cdot = B_\ell \cdot B_\ell^\dag\}_\ell$,
\begin{equation}\label{eq:PEC}
    \mathcal{E}_Q(t, t+\delta t) = \sum_\ell q_\ell(t, \delta t)\mathcal{B}_\ell = \gamma(t, \delta t)\sum_\ell \alpha_\ell(t,\delta t) p_\ell(t,\delta t)\mathcal{B}_\ell. 
\end{equation}
Here the coefficients are defined as follows,
\begin{equation}\label{eq:normalizationconstant}
    \begin{aligned}
    \alpha_\ell(t,\delta t) =& \text{sgn}(q_\ell(t,\delta t)), \\
    p_\ell(t,\delta t) = & \frac{\abs{q_\ell(t,\delta t)}}{\gamma(t,\delta t)},  \\
    \gamma(t,\delta t) = & \sum_\ell|q_\ell(t,\delta t)|. 
\end{aligned}
\end{equation}
In particular, $ \gamma(t,\delta t) $ in \cref{eq:normalizationconstant} is a normalizing factor to ensure that $\sum_\ell p_\ell(t,\delta t) = 1$. More importantly, it is indicative of the overall sampling complexity \cite{sun2021mitigating}.

We list the 16 basis operations \cite{endo2018practical} of one qubit in \cref{BasisTable} for quick reference.

\begin{table}[!h]
\begin{center}
\scriptsize
\begin{tabular}{||c|c|c|c|c|c|c|c||}
\hline
1 &$[I]$
& 2 & $[\sigma^{x}] $
& 3 & $[\sigma^{y}] $
& 4 & $[\sigma^{z}]$
\\ 
5 & $[R_x] = [\tfrac{1}{\sqrt{2}}(1 + i\,\sigma^x)]$
& 6 & $[R_y] = [\tfrac{1}{\sqrt{2}}(1 + i\,\sigma^y)]$
& 7 & $[R_z] = [\tfrac{1}{\sqrt{2}}(1 + i\,\sigma^z)$]
& 8 & $[R_{yz}] = [\tfrac{1}{\sqrt{2}}(\sigma^y + \sigma^z)$]
\\ 
9 & $[R_{zx}] = [\tfrac{1}{\sqrt{2}}(\sigma^z + \sigma^x)$]
& 10 & $[R_{xy}] = [\tfrac{1}{\sqrt{2}}\teal{(\sigma^x + \,\sigma^y)}$]
& 11 & $[\pi_x] = [R_z^3 R_x^3][\pi][R_xR_z]$
& 12 & $[\pi_y] = [R_x][\pi][R_x^3]$
\\ 
13 & $[\pi_z] = [\pi]$
& 14 & $[\pi_{yx}] = [R_z^3 R_x^3][\pi][R_x^3R_z]$
& 15 & $[\pi_{xz}] = [R_x][\pi][R_x^3R_z^2]$
& 16 & $[\pi_{xy}] = [\pi][R_x^2]$
\\ \hline
\end{tabular}
\caption{The standard basis of probabilistic error cancellation in ~\cite[Table 1]{endo2018practical}, including the superoperator notation $[U]\rho := U\rho U^\dag$ and the projection operator $[\pi] = [\ketbra{0}]$. Notice that the last six basis operations $\mathcal{B}_\ell, \ell = 11,\cdots, 16$ contain the projection operator. }
\label{BasisTable}
\end{center}
\end{table}

The basis operations can be easily generalized to multiple qubits by tensor products. 
Overall the operations from $0$ to $T = M\delta t$, as illustrated in \cref{fig:QEM},  are given by,
\begin{equation}\label{eq:nstepQEM}
\begin{aligned}
    \rho_I(T) &= \prod_{k=0}^{M-1} \mathcal{E}_I((k+1)\delta t, k\delta t)\rho(0)\\
    &\approx  \prod_{k=0}^{M-1} \mathcal{E}_Q(k\delta t,(k+1)\delta t)\mathcal{E}_N((k+1)\delta t, k\delta t )\rho(0) \\
    & =  \gamma_{\textbf{tot}}\sum_{\Vec{\ell}}\alpha_{\Vec{\ell}}p_{\Vec{\ell}}\prod_{k=0}^{M-1}\mathcal{B}_{\ell_k}\mathcal{E}_N((k+1)\delta t, k\delta t )\rho(0)=:\rho_Q(T),    
\end{aligned}
\end{equation}
where $\Vec{\ell} = (\ell_0,\cdots, \ell_{M-1})$, $\gamma_{\textbf{tot}} = \prod_{k=0}^{M-1}\gamma(k\delta t, \delta t)$, $\alpha_{\Vec{\ell}} = \prod_{k=0}^{M-1}\alpha_{\ell_k}(k\delta t,\delta t)$ and $p_{\Vec{\ell}} = \prod_{k=0}^{M-1}p_{\ell_k}(k\delta t,\delta t)$. The ultimate goal of QEM is to approximate the expectation of an observable $O$, using Monte Carlo sampling, i.e.,
\begin{equation}\label{eq:QEM_measure}
    \tr(O\rho_I(T)) =  \gamma_{\textbf{tot}}\sum_{\Vec{\ell}}\alpha_{\Vec{\ell}}p_{\Vec{\ell}}\tr(O\prod_{k=0}^{M-1}\mathcal{B}_{\ell_k}\mathcal{E}_N((k+1)\delta t,k\delta t)\rho(0)).
\end{equation}

\begin{figure}[!htp]
    \centering
    
    \begin{subfigure}{0.65\linewidth}
        \centering
        \includegraphics[width=\linewidth]{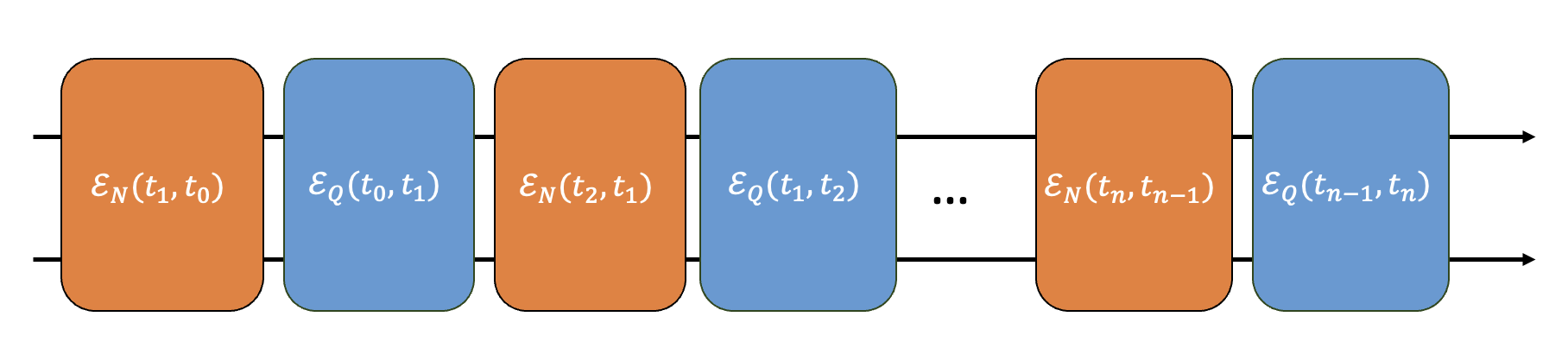}
        \caption{}
        \label{fig:QEM}
    \end{subfigure}

    \begin{subfigure}{0.65\linewidth}
        \centering
        \includegraphics[width=\linewidth]{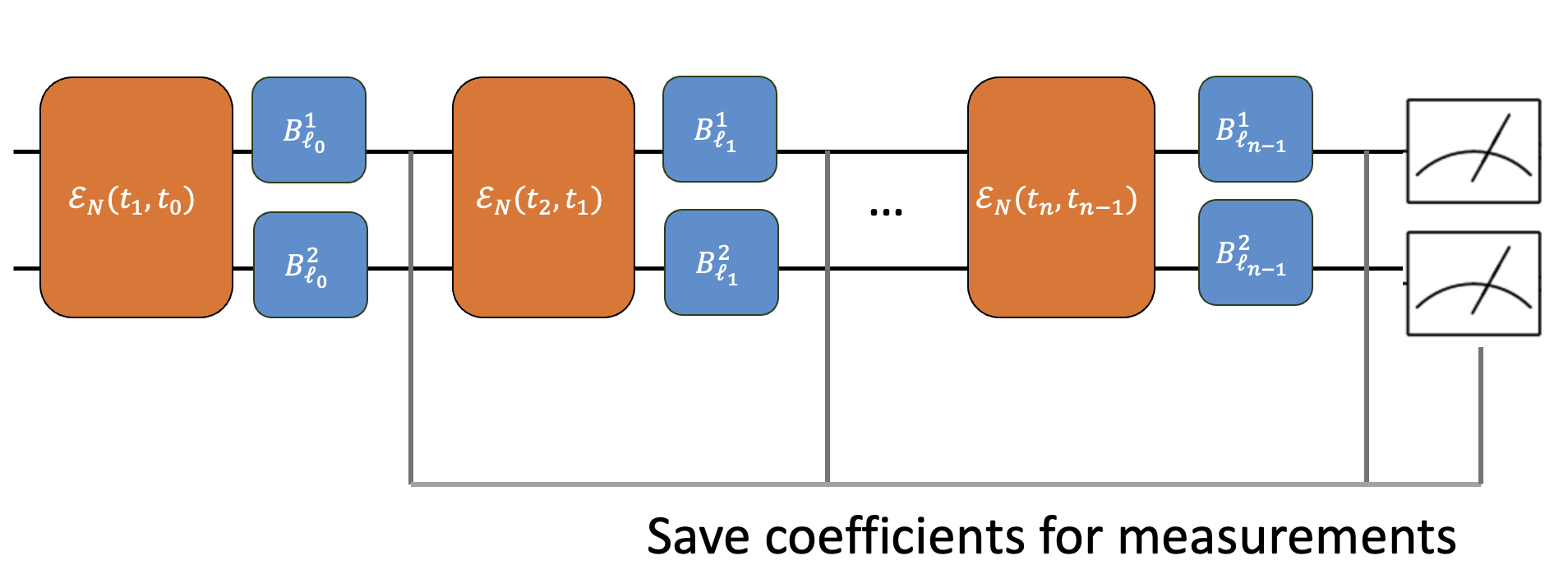}
        \caption{}
        \label{fig:TrajectoryQEM}
    \end{subfigure}

    \caption{ Schematic representation of the multi-step NMNM method in \cref{eq:nstepQEM}. (a) Interlacing structure of the multi-step NMNM method; (b) A sample circuit used in the Monte Carlo sampling implementation. }
\end{figure}

 These standard operation basis are linearly independent and have been designed for local and Markovian noise. Due to their success in Markovian QEM and simplicity, we continue to use {the} basis, while keeping the quasiprobabilities time-dependent. Nevertheless, there are other alternatives, including the time-dependent basis in  \cite{hakoshima2021relationship}, or the Dirac Gram matrices approach by Ahn and coworkers \cite{ahn2023non,ahn2024non}. 
 The quasi-probabilistic expression is also consistent with the general result in \cite{rossini2023single} where the recovery operator is written as the difference between two CP maps.

\section{Non-Markovian Noise Mitigation}\label{NMEM}
In this section, we present the noise mitigation framework tailored for non-Markovian open quantum systems, extending the continuous PEC method in~\cref{Preliminary}. The key steps are deriving an appropriate non-Markovian noise model and constructing the recovery operation properly. We introduce the underlying master equation formalism in \cref{masterEqNM}, and one-step NMNM in \cref{1stepNMNM} and subsequently analyze the bias and sampling overhead of the multiple-step scenario in \cref{AnalysisCost}. 

\subsection{Time-Local Master Equation for Non-Markovian Noise}
\label{masterEqNM}
The dynamics of a quantum system under the influence of an environment is often referred to as an open quantum system \cite{breuer2002theory}. Mathematically, the study of an open quantum system often starts with  a composite Hamiltonian of the form \cite{breuer2002theory}
\begin{equation}\label{Htot}
    H_\text{tot} = H_S \otimes I_B + I_S \otimes H_B + \lambda H_{SB},
\end{equation}
where $H_S$ and $H_B$ are respectively the system and bath Hamiltonians, and $\lambda$ is the coupling strength between the system and the environment. In addition,  the total system evolves according to the von Neumann equation
\begin{equation}\label{lvn}
    \partial_t\rho_\text{tot}(t) = -i[H_\text{tot}, \rho_\text{tot}],\quad \rho_\text{tot}(0) = \rho_S(0)\otimes \rho_B.
\end{equation}
Here $\rho_B$ is the thermal state 
$\rho_B = \frac{e^{-\beta H_B}}{Z_B}$  
with parameter $\beta = 1/\kappa_B T$ and $\kappa_B$ is the Boltzmann constant. The normalization factor $Z_B= \Tr(e^{-\beta H_B})$ is the reservoir partition function.

We also follow the standard assumption that the interaction term is given by \cite{breuer2002theory,carmichael2013statistical}, 
\begin{equation}
    H_{SB} = \sum_{j = 1}^J S_j\otimes B_j.
\end{equation}
Without loss of generality,  $S_j$ and $B_j$ can be assumed to be Hermitian. In addition, by proper shifting, we can assume that $\tr(B_j \rho_B)=0$ \cite{carmichael2013statistical}.  

A key determining factor in the dynamics of the open quantum system is the bath correlation function (BCF), e.g., as indicated by the Feynman-Vernon influence functional representation~\cite{breuer2004time,tamascelli2018nonperturbative,jin2008exact}. With $B_j(t) = e^{itH_B }B_j e^{-it H_B }$, the two-point BCF is defined as follows,
\begin{equation}
    C_{j,k}(t) = \tr\left(B_j(t) B_k(t) \rho_B\right).
\end{equation}
It satisfies the symmetry property $C_{j,k}(t) = C_{k,j}(-t)^*, t<0$. For example, for Gaussian environment, it completely determines the dynamics of non-Markovian open quantum system~\cite{tamascelli2018nonperturbative}.

To allow explicit derivations, we express the BCF in the following form,
\begin{equation}\label{eq: Cjkt}
     C_{j,k}(t) = \sum_{\mu}g^*_{j,\mu}g_{k,\mu}e^{i\omega_\mu t},\quad\omega_\mu\in\mathbb{C}, t>0.
\end{equation}
 It is important to notice that for continuous bath, the BCF is often expressed as a Fourier integral, which can be reduced to \cref{eq: Cjkt} by a quadrature formula or pole expansion \cite{ritschel2014analytic}. In the latter case, the frequency often takes complex values with a positive imaginary part.

By a direct perturbation analysis of \eqref{lvn}, i.e., by seeking the solution as 
\( \rho_\text{tot} = \rho_\text{tot}^{(0)}+\lambda  \rho_\text{tot}^{(1)}+ \lambda^2 \rho_\text{tot}^{(2)} + \cdots  \), 
followed by tracing out the environment, the dynamics of the system can be expressed as follows,
\begin{equation}\label{eq: M2}
\begin{aligned}
        \rho_S(t) &=U_S(t)\rho_S(0)U_S^\dag(t)+\lambda^2\sum_{j,k=1}^{J}  U_S(t)   \int_0^t \int_0^t {S}_{j}(t_1) \rho_S(0)  {S}_{k}(t_2)  C_{k,j}(t_2-t_1) dt_2dt_1 U_S^\dag(t)\\
  &-U_S(t)\int_0^t\int_0^{t_1}\rho_S(0) {S}_{k}(t_2){S}_{j}(t_1)C_{k,j}(t_2-t_1)\mathrm{d}t_2\mathrm{d}t_1U_S^\dag(t)\\
  &-U_S(t)\int_0^t\int_0^{t_1}{S}_{j}(t_1){S}_{k}(t_2)\rho_S(0) C_{j,k}(t_1-t_2)\mathrm{d}t_2\mathrm{d}t_1U_S^\dag(t) +O(\lambda^4t^4),\\    
\end{aligned}
\end{equation}
where $U_S(t):=e^{-itH_S }$ and $S_j(\tau) = e^{iH_S\tau }S_j e^{-iH_S \tau} $ is the dynamics of the operator $S_j$ in the Heisenberg picture.

By defining the new jump operators
\begin{equation}\label{That}
    \widehat{T}_\mu = \sum_j g_{j,\mu} S_j,
\end{equation}
 and calculating the derivative of the evolution operator, we obtain a time-local quantum master equation
\begin{equation}\label{eq: nM}
    \begin{aligned}
       \partial_t\rho_S(t) 
    &= -i[H_S,\rho_S(t)]+\lambda^2\sum_{\mu = 1}^{\mu_{\max}} \int_0^t\mathcal{L}(\widehat{T}_\mu , \widehat{T}_\mu(-t_2)e^{i\omega_{\mu}t_2})\mathrm{d}t_2
    \rho_S(t)+O(\lambda^4) \\
    &= -i[H_S,\rho_S(t)]+\mathcal{L}_N(t)\rho_S(t)+O(\lambda^4).
    \end{aligned}
\end{equation} 
Again we used the notation $\widehat{T}_\mu(\tau) = {e^{iH_S\tau }\widehat{T}_\mu e^{-iH_S \tau} }$ defined in the Heisenberg picture. In addition, we have defined the time-dependent noise operator that embodies a memory effect,
 \begin{equation}\label{bcf2LN}
     \mathcal{L}_N(t) \rho(t)= \lambda^2\sum_{\mu = 1}^{\mu_{\max}} \int_0^t\mathcal{L}(\widehat{T}_\mu , \widehat{T}_\mu(-t_2)e^{i\omega_{\mu}t_2})\mathrm{d}t_2
    \rho(t),
 \end{equation}
 where 
\begin{equation}\label{LFG}
    \mathcal{L}(F,G)\rho := F\rho G^\dag+G\rho F^\dag-\rho G^\dag F-F^\dag G\rho.
\end{equation}

We have derived a quantum master equation (QME) in the weak-coupling regime that will serve as the foundation for our quantum error mitigation schemes. Notably, unlike the derivation in \cite{hall2014canonical}, the QME in \cref{eq: nM} does not require the invertibility of a dynamical map {(the notion of invertibility is slightly different from that in \cite{jagadish2023noninvertibility}).} More importantly, this derivation reveals a direct connection to the bath correlation functions (BCFs), especially the distribution of their poles, as illustrated in \cref{bcf2LN}.

\subsection{One-step Noise Mitigation Algorithm}
\label{1stepNMNM}

The noisy and recovery operator in $[t, t+\delta t]$ of our NMNM method are defined as
\begin{equation}\label{eq:NMNM}
    \begin{aligned}
        & \mathcal{E}_N(t+\delta t, t)\rho(t) = \cE_I(t+\delta t, t)\rho(t) + \int_t^{t+\delta t}\mathcal{U}_S(t+\delta t-\tau)\mathcal{L}_N(\tau)\rho_N(\tau)\mathrm{d}\tau,\\
        & \mathcal{E}_Q(t,t+\delta t)\rho(t+\delta t) = (I-\delta t\mathcal{L}_N(t+\delta t))\rho(t+\delta t),
    \end{aligned}
\end{equation}
where the ideal operator $\cE_I(t+\delta t, t)\rho(t) = e^{-iH_S\delta t}\rho(t)e^{iH_S\delta t} = \bl{\mathcal{U}_S(\delta t)\rho(t)}$. It is worth mentioning that the recovery operator at time interval $[t, t+\delta t]$ depends on the integration over the full interval $[0,t+\delta t]$, underscoring the non-Markovian property. With further manipulations of the last equation, one can show that this recovery operator offers a first-order approximation to the ideal evolution $\cE_I(t+\delta t, t)\rho(t)$, which will be proved in \cref{lm: 1stepQEM}.

We proceed
to map the recovery operator to circuit operations. We 
 first expand the matrix $\widehat{T}_\mu$, and similarly the operator  $\int_0^t\widehat{T}_\mu(-\tau)e^{i\omega_\mu \tau}\mathrm{d}\tau$ from \cref{bcf2LN} to an orthogonal basis $\{V_\alpha\}_{\alpha}$ in $\mathbb{C}^{N\times N}  (N = 2^n)$, 
\begin{equation}
    \begin{aligned}
        \widehat{T}_\mu = \sum_{\alpha=1}^M f_\alpha^\mu V_\alpha,\quad \int_0^t\widehat{T}_\mu(-\tau)e^{i\omega_\mu \tau}\mathrm{d}\tau = \sum_\alpha \int_0^t\bl{v_\alpha^\mu}(\tau)\mathrm{d}\tau V_\alpha,
    \end{aligned}
\end{equation}
with the coefficients given by $f_\alpha^\mu = \Tr(\widehat{T}_\mu V_\alpha^\dag)/\tr(V_\alpha V_\alpha^\dag)$ and $ \bl{v_\alpha^\mu}(\tau) = \Tr(\widehat{T}_\mu(-\tau)e^{i\omega_\mu \tau}V_\alpha^\dag)/\tr(V_\alpha V_\alpha^\dag)$. 

Denote $A_{\alpha, \beta}(t) = \sum_\mu f^\mu_\alpha \int_0^t \bl{v^{\mu*}_\beta}(\tau)\mathrm{d}\tau$, then the recovery operator in \cref{eq:NMNM} can be rewritten as  
\begin{equation}\label{eq:ErecoveryA}
    (I-\delta t\mathcal{L}_N(t+\delta t))\rho = \rho-\delta t\lambda^2 \sum_{\alpha,\beta} A_{\alpha,\beta}(t+\delta t)(V_\alpha \rho V_\beta^\dag -\rho V_\beta^\dag V_\alpha)+{A}_{\beta,\alpha}^*(t+\delta t) (V_\alpha \rho V_\beta^\dag- V_\beta^\dag V_\alpha\rho),
\end{equation}
which begins to show some resemblance with a non-diagonal Lindblad operator. 
The coefficient matrix $A(t) = (A_{\alpha\beta}(t))_{\alpha\beta}$ can be further decomposed into the Hermitian and the skew Hermitian components
\begin{equation}
    A(t) = \Gamma(t)+ i\Xi(t), \quad \Gamma^{\dag}(t) = \Gamma(t), \;  \Xi^{\dag}(t) = \Xi(t).
\end{equation}
This separates the recovery operator into the {coherent and incoherent} parts, 
\begin{equation}\label{eq:LDLC}
{\mathcal{L}_C(t)\rho =\lambda^2i [ \sum_{\alpha,\beta} \Xi_{\alpha\beta}(t)V_\beta^\dag V_\alpha, \rho ],\quad \mathcal{L}_D(t)\rho = \lambda^2 \sum_{\alpha,\beta} \Gamma_{\alpha\beta}(t)(2V_\alpha\rho V_\beta^\dag -\rho V_\beta^\dag V_\alpha - V_\beta^\dag V_\alpha \rho).}
\end{equation}
such that $\delta t\mathcal{L}_N(t)\rho = \delta t\mathcal{L}_C(t)\rho + \delta t\mathcal{L}_D(t)\rho$.

With this simplification, the time-local quantum master equation exhibits the following compact form,
\begin{equation}\label{eq: QME}
       {\partial_t\rho_S(t) 
    = -i[H_S- \lambda^2 \Delta_S,\rho_S(t)]+\lambda^2 \sum_{\alpha,\beta} \Gamma_{\alpha\beta}(t)(2V_\alpha\rho V_\beta^\dag -\rho V_\beta^\dag V_\alpha - V_\beta^\dag V_\alpha \rho).}
\end{equation} 
The extra term in the Hamiltonian, $ \lambda^2 \Delta_S$, can be recognized as the { Lamb-shift}. Meanwhile, the remaining term resembles the non-diagonal Lindblad operator. However,  if an eigenvalue of $\Gamma_{\alpha\beta}(t)$ becomes negative, then the dynamical map associated with \cref{eq: QME} can no longer be divisible into completely positive maps \cite{breuer2009measure}, thus giving rise to non-Markovian dynamics. 

We now return to the recovery operator $\mathcal{E}_Q$, which in light of \cref{eq: QME}, should be designed to offset the effect of $\mathcal{L}_C(t)$ and $\mathcal{L}_D(t)$.  By operator splitting scheme, we can implement the {coherent and incoherent} part of the recovery operator $\mathcal{E}_Q$ separately, i.e.,
${e^{-\delta t\mathcal{L}_N(t)}=e^{-\delta t\mathcal{L}_C(t)} e^{-\delta t\mathcal{L}_D(t)} + O(\delta t^2). }$ Then, the quasi-probability can be calculated via \cref{eq:PEC} and \cref{eq:normalizationconstant}.

\subsection{Error and Complexity Analysis }\label{AnalysisCost}
With a stochastic implementation of our QEM scheme, we arrive at the following stochastic circuit operator,
\[
 \tr\bigl(O \cE_{I}(t,0) \rho_0\bigr) \approx \frac{1}{N_r} \sum_{m=1}^{N_r} 
  \tr\bigl(O      \prod_{k=0}^{M-1} \mathcal{B}_{\ell_k}^{(m)} \cE_{N}((k+1)\dt,k\dt) \rho_0\bigr). 
\]

To ensure that the estimated expectation has precision $\epsilon$ with high probability, we first choose $\dt$ so that the bias is within $\epsilon$. Namely, we can enforce 

\begin{equation}\label{bias}
    \bl{\sum_{\Vec{\ell}} \alpha_{\Vec{\ell}}}\, \mathbb{E}\left[   \prod_{k=0}^{M-1}  \mathcal{B}_{\ell_k} \cE_{N}((k+1)\dt,k\dt) \rho_0  \right] -  \ \cE_{I}(t,0) \rho_0 = O(\epsilon).
\end{equation}

Meanwhile, to estimate the statistical error, we can invoke Hoeffding's concentration inequality. Toward this end,  let us define 
\begin{equation}
   o^{(m)}=  \tr\bigl(O      \prod_{k=0}^{M-1}  \mathcal{B}_{\ell_k}^{(m)} \cE_{N}((k+1)\dt,k\dt) \rho_0\bigr), 
\end{equation}
be the outcome from the $m$th random circuit and $\Bar{o}$ be the expectation value. Then,
\begin{equation}\label{eq:propSampling}
    \mathbb{P}\left( \abs{\frac{1}{N_r} \sum_{m=1}^{N_r}o^{(m)} - \Bar{o} }>\epsilon \right) < 2 e^{ -\frac{N_r \epsilon^2}{ \gamma_\textbf{tot}} },
\end{equation}
where the total normalizing factor is given by,
\begin{equation}\label{gamma_tot}
    \gamma_\textbf{tot} = \prod_{k=0}^{M-1}\gamma(k\dt, \dt),
\end{equation}
with $\gamma(k\dt, \dt)$  defined in \eqref{eq:normalizationconstant}.

\bigskip 
 We start with the one-step error for the recovery operator {\eqref{1step-error}} in \cref{lm: 1stepQEM} and move forward to the error analysis of the multiple-step scenario afterward. Throughout this analysis, $\|\cdot\|$ and $ \|\cdot\|_1$ denotes the trace norm and element-wise matrix 1-norm.  By rescaling, we assume $\norm{S_j}=1, \forall j\in [J]$ without loss of generality.

\begin{lemma}\label{lm: 1stepQEM}
   For any $t$ and $\delta t>0$, the recovery operator $\mathcal{E}_Q(t, t+\delta t)$ defined in ~\cref{eq:NMNM} produces a one-step error that can be bounded by,
    \begin{equation}\label{1step-error}
        \norm{ \mathcal{E}_Q(t, t+\delta t)\mathcal{E}_N(t+\delta t, t) \rho(t) - \mathcal{E}_I(t+\delta t, t)\rho(t)  } \leq \delta t^2\lambda^2 (\|H_S\|G_{b,1} + G_{b,2}  e^{-\theta t}),
    \end{equation}
   where the constants are given by, 
   \begin{equation}\label{Gb12}
       G_{b,1} = 2\sum_\mu(\sum_{j}|g_{j,\mu}|)^2\frac{1}{\Im(\omega_\mu)},\quad G_{b,2} = \frac{1}{2}\sum_\mu(\sum_{j}|g_{j,\mu}|)^2,\quad \theta = \inf_\mu\Im(\omega_\mu).
   \end{equation}
  Since $G_{b,1}\leq 4G_{b,2}/\theta$, a slightly simpler one-step error bound is given by,
  \begin{equation}\label{1step-error2}
      \norm{ \mathcal{E}_Q(t, t+\delta t)\mathcal{E}_N(t+\delta t, t) \rho(t) - \mathcal{E}_I(t+\delta t, t)\rho(t)  } \leq \delta t^2\lambda^2 G_{b,2} (4\|H_S\|/\theta +  e^{-\theta t}).
  \end{equation}
\end{lemma}
\begin{proof}
    From \cref{eq:NMNM}, the one-step NMNM recovers the quantum state to
    \begin{equation}
         \mathcal{E}_Q(t,t+\delta t)\mathcal{E}_N(t+\delta t, t)\rho(t) = (I-\delta t\mathcal{L}_N(t+\delta t))\mathcal{E}_N(t+\delta t, t)\rho(t).
    \end{equation}
    
    Consequently, the difference between the error-mitigated density operator and the ideal density operator becomes
    \begin{equation}
    \begin{aligned}
        &\bl{\mathcal{E}_Q(t,t+\delta t)}\mathcal{E}_N(t+\delta t, t)\rho(t) - \cE_I(\delta t)\rho(t)= -\delta t \mathcal{L}_N(t+\delta t)\rho_N(t+\delta t) + \int_t^{t+\delta t}\mathcal{U}_S(t+\delta t-\tau)\mathcal{L}_N(\tau)\rho_N(\tau)\mathrm{d}\tau \\       
    \end{aligned}
    \end{equation}

The operator here is expressed as an integral. We will use the simple rectangle rule to approximate it and provide an error bound. Toward this end,  let the function
\[ f(\tau)  = \mathcal{U}_S(t+\delta t-\tau)\mathcal{L}_N(\tau)\rho_N(\tau), \] 
with its derivative given by,
    \begin{equation}
    f'(\tau) = i[H_S,f(\tau)]+\mathcal{U}_S(t+\delta t-\tau)\left(\lambda^2\sum_\mu \mathcal{L}(\widehat{T}_\mu, \widehat{T}_\mu(-\tau)e^{i\omega_\mu \tau})\rho_N(\tau)-i\mathcal{L}_N(\tau)[H_S,\rho_N(\tau)]\right).
    \end{equation}
    
    The error of the quadrature depends on the derivative bound. Using \cref{That,LFG} and
    $\theta = \inf_\mu\Im(\omega_\mu)$,  we have
    \begin{equation}
    \begin{aligned}
    \|f'(\xi)\|&\leq \lambda^2\left(2\|H_S\|\sum_\mu \|\widehat{T}_\mu\|^2\int_0^{t+\delta t} e^{-\Im(\omega_\mu) t_2}\mathrm{d}t_2+\sum_\mu\|\widehat{T}_\mu\|^2 e^{-\Im(\omega_\mu) t}\right)\\
    &\leq 4\lambda^2 \|H_S\| \sum_\mu\sum_{j,k}|g_{j,\mu}||g_{k,\mu}|\frac{1}{\Im(\omega_\mu)} + \lambda^2 \sum_\mu\sum_{j,k}|g_{j,\mu}||g_{k,\mu}|e^{-\theta t}.
    \end{aligned}
    \end{equation}
    Then by using quadrature formula, we arrive at,
    \begin{equation}
        \|\bl{\mathcal{E}_Q(t,t+\delta t)}\mathcal{E}_N(t+\delta t, t)\rho(t) - \cE_I(\delta t)\rho(t)\| \leq \frac{1}{2}\delta t^2 \max_{\xi\in[t,t+\delta t]}|f'(\xi)|\leq \delta t^2\lambda^2 (\|H_S\|G_{b,1} + G_{b,2}  e^{-\theta t}). 
    \end{equation}
\end{proof}
\begin{theorem}\label{thm:bias}
{Let $\rho_Q(T)$ be the expectation of the density operator from the stochastic QEM in \cref{eq:nstepQEM} at the end time $T = M\delta t$. } It produces an approximation of the the ideal density operator $\rho_I(T)$ with the following error bound,
    \begin{equation}\label{rho-bound}
        \norm{\rho_I(T) - \rho_Q(T)} \leq \delta t T\lambda^2 \|H_S\|G_{b,1}+\delta t^2\lambda^2 G_{b,2}\frac{1}{1-e^{-\theta \delta t}},
    \end{equation}
    with the same coefficients defined in \cref{Gb12}. 
\end{theorem}
\begin{proof} 
The $M$-step QEM error can be upper-bounded by the triangle inequality
\begin{equation}\scriptstyle
    \|\rho_Q(T)-\rho_I(T)\|\leq \sum_{k=1}^M\|\mathcal{E}_I(T,k\delta t)\mathcal{E}_Q((k-1)\delta t, k\delta t)\mathcal{E}_N(k\delta t, (k-1)\delta t)\rho_Q((k-1)\delta t)-\mathcal{E}_I(T,(k-1)\delta t)\rho_Q((k-1)\delta t)\|.
\end{equation}
By substituting the right-hand side with the 1-step QEM error bound from \cref{lm: 1stepQEM}, we find that, 
\begin{equation}
    \|\rho_Q(T)-\rho_I(T)\|\leq \delta t T\lambda^2 G_{b,1}+\delta t^2\lambda^2 G_{b,2}\frac{1}{1-e^{-\theta \delta t}}.
\end{equation}
\end{proof} 

In light of  \cref{eq: QME}, the error caused by the \teal{Lamb-shift} can be directly simulated by Hamiltonian simulation.  It is the incoherent part that requires the quasi-probabilistic approach. Therefore, focus our analysis on the mitigation of $\mathcal{L}_D$ in \cref{eq:LDLC}. 
\begin{theorem}[Error Mitigation for the {incoherent} part of the non-Markovian noise]\label{thm:gamma}
    The normalization coefficient in~\cref{gamma_tot} of our NMNM method for the non-Markovian {incoherent noise} has the following bound
    \begin{equation}\label{gamma-bound}
       \gamma_\textbf{tot} \leq e^{O(\lambda^2 TG_{b,1})} ,
    \end{equation}
    with the constants defined in \cref{Gb12}.
\end{theorem}

Remarkably, the parameter $G_{b,1}$ appeared in the complexity bound as well. 

\begin{proof}
We consider the QEM scheme for the {incoherent} part of the non-Markovian noise, i.e.
\begin{equation}
   \mathcal{E}_Q(t-\delta t, t)\rho(t) = \rho(t)-\delta t\lambda^2 \sum_{\alpha,\beta} \Gamma_{\alpha\beta}(t)(2V_\alpha\rho(t) V_\beta^\dag -\rho(t) V_\beta^\dag V_\alpha - V_\beta^\dag V_\alpha \rho(t)),\quad \Gamma_{\alpha,\beta} = \frac{1}{2}(A_{\alpha,\beta}+A_{\beta,\alpha}^*),
\end{equation}
where
\begin{equation}
    2V_\alpha\rho(t) V_\beta^\dag -\rho(t) V_\beta^\dag V_\alpha - V_\beta^\dag V_\alpha \rho(t) = \sum_\ell d_{\ell,\alpha,\beta}\mathcal{B}_\ell \rho(t), \quad d_{\ell,\beta,\alpha} = d_{\ell,\alpha,\beta}^*.
\end{equation}
A direct substitution yields,
\begin{equation}
    \mathcal{E}_Q\teal{(t-\delta t,t)}\rho(t) =(1-\delta t\lambda^2 \sum_{\alpha,\beta}\Gamma_{\alpha\beta}(t)d_{0,\alpha,\beta})\rho(t)+\delta t\lambda^2 \sum_{\ell>0}\sum_{\alpha,\beta}\Gamma_{\alpha\beta}(t)d_{\ell,\alpha,\beta}\mathcal{B}_\ell\rho(t),
\end{equation}
and the normalization coefficient from \cref{eq:normalizationconstant} is thus given by,
\begin{equation} \label{c2gamma}
    \gamma(t,\delta t) := 1-\delta t\lambda^2 \sum_{\alpha,\beta}\Gamma_{\alpha\beta}(t)d_{0,\alpha,\beta} + \delta t\lambda^2\sum_{\ell>0}|\sum_{\alpha,\beta}\Gamma_{\alpha\beta}(t)d_{\ell,\alpha,\beta}|.
\end{equation}
Notice that for any $\ell\geq 0 $,
    \begin{equation}\label{gamma2c}
        \sum_{\alpha,\beta}\Gamma_{\alpha,\beta}d_{\ell,\alpha,\beta} = \sum_{j,k}\int_0^tC_{j,k}^*(\tau)D_{j,k,\ell}^*(\tau) + C_{k,j}(\tau)D_{k,j,\ell}(\tau)\mathrm{d}\tau,
    \end{equation}
    where
    \begin{equation}
        D_{j,k,\ell}(\tau)  = \sum_{\alpha,\beta} \tr(S_jV_\alpha) \tr(S_k(-\tau)V_\beta)d_{\ell,\beta,\alpha}.
    \end{equation}
    By defining a new matrix 
    $G^{(\ell)} = (G_{j,k}^{(\ell)})$ as the Hermitian part of the matrix $ (C_{j,k}D_{j,k,l})_{j,k}$, then
    \begin{equation}
    \sum_{\alpha,\beta}\Gamma_{\alpha,\beta}d_{\ell,\alpha,\beta}  = \sum_{j,k}\int_0^t G_{j,k}^{(\ell)}\mathrm{d}\tau
    \end{equation}

    This simplified \cref{c2gamma} to
    \begin{equation}
        \gamma(t,\delta t) = |1-\delta t \lambda^2 \sum_{j,k}\int_0^t G_{j,k}^{(0)}(\tau)\mathrm{d}\tau| +\delta t\lambda^2 \sum_{\ell>0} |\sum_{j,k}\int_0^t G_{j,k}^{(\ell)}(\tau)\mathrm{d}\tau|.
    \end{equation}
       Observe that $\norm{S_j(-\tau)} \leq 1$ for any $\tau $, $V_\alpha$ is a fixed basis, and $d_{\ell,\beta,\alpha}$ is also independent of the noise operator.   So there exists a constant $Q$ such that $\sum_{\ell\geq 0}\max_{j,k}\abs{D_{j,k,\ell}(\tau)}\leq Q$. 
       Therefore, we can use the Cauchy-Schwarz inequality to the product of $C$ and $D$, 
    \begin{equation}
        \gamma(t,\delta t) \leq  1 + 2\delta t\lambda^2 Q\int_0^t\|C(\tau)\|_1\mathrm{d}\tau,
    \end{equation}
    where the element-wise 1-norm of BCF $\|C(\tau)\|_1 = \sum_{j,k}|C_{j,k}(\tau)|\leq  \sum_{j,k}\sum_{\mu}|g_{j,\mu}||g_{k,\mu}|e^{-\beta_\mu \tau}$. 
    The normalization coefficient of multi-step quantum error mitigation
    \begin{equation}\label{eq:overhead0}
    \begin{aligned}
            \gamma_\textbf{tot} &=\prod_{k=0}^{M-1}\gamma(k\delta t,\delta t) \leq \prod_{k=0}^{M-1}(1+2\delta t\lambda^2 Q\int_0^{(k+1)\delta t}\|C(\tau)\|_1\mathrm{d}\tau)\leq \exp(2\lambda^2 QT\sum_{j,k,\mu}|g_{j,\mu}||g_{k,\mu}|\frac{1}{\Im(\omega_\mu)})\\
            &\leq \exp(O(\lambda^2 TG_{b,1})).
    \end{aligned}
    \end{equation}
\end{proof}

Combining the previous results with Hoeffding's inequality, we summarize the complexity of NMNM  as follows. 
\begin{theorem}[Sampling Complexity]\label{thm:complexity}
    Given an observable $O$ and accuracy parameters $\epsilon,\delta \in (0,1]$. Set
    \begin{equation}
       \delta t = O\left(\frac{\epsilon}{\lambda^2 (T \|H_S\|G_{b,1} + G_{b,2}\frac{1}{1-e^{-\theta}})}\right),\quad N_r  = \Omega\left(\log(1/\delta )( \exp(\lambda^2 T G_{b,1})/\epsilon^2)\right).
   \end{equation}
   Then, $N_r$ samples of NMNM trajectories with outcomes $o^{(m)}, m = 1,\cdots, N_r$ at time $T = M\delta t$ allow for an accurate approximation of the ideal measurement such that
   \begin{equation}
       \left|\tr(O\rho_I(T))-\frac{1}{N_r} \sum_{m=1}^{N_r}o^{(m)}\right|\leq \epsilon,
   \end{equation}
   with probability at least $1-\delta $.
\end{theorem}
\begin{proof}
    Notice that the average value $\Bar{o} = \tr(O\rho_Q(T)) =\mathbb{E}_m[o^{(m)}]$. By triangle inequality 
    \begin{equation}
        \mathbb{P}\left( \abs{\frac{1}{N_r} \sum_{m=1}^{N_r}o^{(m)} - \tr(O\rho_I(T)) }>\epsilon \right)\leq \mathbb{P}\left( \abs{\frac{1}{N_r} \sum_{m=1}^{N_r}o^{(m)} - \Bar{o} }>\epsilon/2 \right) + \mathbb{P}\left( \abs{ \Bar{o} - \tr(O\rho_I(T))}>\epsilon/2 \right).
    \end{equation}
    Substitute the result from \cref{thm:gamma} into \cref{eq:propSampling}, we have
    \begin{equation}
        \mathbb{P}\left( \abs{\frac{1}{N_r} \sum_{m=1}^{N_r}o^{(m)} - \Bar{o} }>\epsilon/2 \right)\leq \exp(-\frac{N_r\epsilon^2}{4\exp(O(\lambda^2 T G_{b,1}))}).
    \end{equation}
    From the result in \cref{thm:bias} and the inequality for the exponential function $e^{-\theta \delta t}\leq 1-\delta t(1-e^{-\theta})$, the probability of the bias is bounded by,
    \begin{equation}
         \mathbb{P}\left( \abs{ \Bar{o} - \tr(O\rho_I(T))}>\epsilon/2 \right) \leq 1-\mathbb{P}\left(\delta t\lambda^2 (T\|H_S\| G_{b,1} + G_{b,2}\frac{1}{1-e^{-\theta}})\leq \epsilon/2\right).
    \end{equation}
    
    Consequently, to obtain the ideal measurement with $\epsilon$-precision, the time step $\delta t$ and the number of samples are of the order stated in the theorem.
\end{proof}

\section{Numerical Experiments}\label{NumericalExperiment}
We consider a prototypical model of non-Markovian noise, i.e., the spin-boson model with 1 or 2 spins (qubits) coupled to a common bath. 
Although the separation of incoherent and coherent parts in \cref{AnalysisCost} is insightful, in our numerical implementation for non-Markovian noise error mitigation, it is not required. We choose the Pauli basis as the orthogonal basis $\{V_\alpha\}_\alpha$. 
To obtain the quasi-probability distribution in~\cref{eq:normalizationconstant}, it is sufficient to project the superoperator in~\cref{eq:ErecoveryA}, i.e.
\begin{equation}
    \begin{aligned}
        V_\alpha \rho V_\beta^\dag -\rho V_\beta^\dag V_\alpha  = \sum_{\ell=0}^{M-1} u_{\ell,\alpha,\beta}^1\mathcal{B}_{\ell}\rho,\quad
        V_\alpha \rho V_\beta^\dag- V_\beta^\dag V_\alpha\rho  = \sum_{l=0}^{M-1} u_{\ell,\beta ,\alpha}^{1*}\mathcal{B}_{\ell}\rho\\ 
    \end{aligned}
\end{equation}
then we can apply each step of the recovery operator through the basis 
\begin{equation}\label{E-Q}
     \mathcal{E}_Q(t, t+\delta t) \rho = \sum_\ell q_\ell(t)\mathcal{B}_\ell\rho. 
\end{equation}
In particular, the coefficients are given by, 
\begin{equation}\label{eq:quasi-prob}
    q_\ell(t) = \begin{cases}
         1-\delta t\lambda^2\sum_{\alpha,\beta}F_{\alpha,\beta}^{(0)}(t)& \ell = 0\\
         -\delta t\lambda^2\sum_{\alpha,\beta}F_{\alpha,\beta}^{(\ell)}(t) & \ell > 0.\\
    \end{cases}
\end{equation}
Here $F_{\alpha,\beta}^{(\ell)}(t) = A_{\alpha,\beta}(t)u_{\ell,\alpha,\beta}^1 +{A}^*_{\beta, \alpha}(t)u_{\ell,\beta ,\alpha}^{1*}, \ell\geq 0 $. 

In order to test the effectiveness of a mitigation method on a quantum simulator, one must first simulate the non-Markovian noise channel  $\cE_N$. This is not straightforward because many non-Markovian models \cite{tanimura2020numerically,li2021markovian} require additional degrees of freedom to capture the memory effect accurately. To avoid using extra qubits and to correctly account for noise trajectories, we simulate the noisy channel via a non-Markovian unraveling method.  Specifically, the noisy evolution $\cE_N$ in \cref{eq: nM} is simulated by solving the non-Markovian stochastic Schr\"odinger equation (NMSSE)~\cite{gaspard1999non,biele2012stochastic}. 
\begin{equation}\label{eq:NMSSE}
    \partial_t\psi  = -i{H}_S\psi -\lambda^2\int_0^t \sum_{j,k}C_{j,k}(\tau) S_j e^{-i{H}_S\tau}S_k\psi(t-\tau)\mathrm{d}\tau +\lambda\sum_j\eta_j(t) S_j\psi(t).
\end{equation}
Here $\eta_j(t)$ is a Gaussian noise with mean zero and covariance given by 
the BCF $C_{j,k}.$ Given the BCF, $\eta_j(t)$  can be sampled using fast Fourier transform. Meanwhile, the integral term can be directly computed by a standard quadrature formula.

\begin{algorithm}
\caption{Quantum Error Mitigation for non-Markovian Noise}\label{psudocode}
\begin{algorithmic}[1]
\BState $\ket{\psi} = \ket{\psi_0}$, $\text{coefficient} = 1$, sampled index $\Vec{\ell} = (\ell_k)$
\For{ $k = 0:n-1$}
\State Noisy Evolution $\cE_N$: $\ket{\psi_{k+1}}\gets \ket{\psi_0}, \cdots \ket{\psi_k}$ evolves from NMSSE~\cref{eq:NMSSE} at $[k\delta t, (k+1)\delta t]$
\State Error Mitigation Operator: $\ket{\psi_{k+1}} \gets B_{\ell_k}\ket{\psi_{k+1}}$, apply projective measurement if necessary
\State Save Coefficients: $\text{coefficient} \gets  \text{coefficient}\times \gamma( k\delta t, \delta t)\times \alpha_{\ell_k}(k\delta t, \delta t)$ 
\State Measurements: $\langle O\rangle  = \bra{\psi_{k+1}} O\ket{\psi_{k+1}} \times \text{coefficient}$
\BState \textbf{end}
\EndFor
\end{algorithmic}
\end{algorithm}

The coefficients $u_{\ell,\alpha,\beta}^1$ are obtained by solving the linear systems constructed by the Pauli transfer matrix as in \cite[Appendix D]{endo2018practical} and \cite{greenbaum2015introduction}. Furthermore, the recovery operator $\mathcal{E}_Q$ is implemented $B_{\ell_k}\ket{\psi}$ with sampled index ${\ell_k}$ in terms of the probability distribution in \cref{eq:QEM_measure}. The algorithm of NMNM is summarized in the pseudocode in \cref{psudocode} with a workflow illustrated in \cref{fig:TrajectoryQEM}. Notice that the coefficient at each layer should be saved by the $\text{coefficient}$ variable to obtain the measurements in \cref{eq:QEM_measure}.
Additionally, the last six superoperators $\mathcal{B}_\ell$ in ~\cref{BasisTable} are not unitary and are usually implemented by splitting the corresponding matrix $B_\ell = U_{\ell,1}\pi U_{\ell,2}$ where $\pi$ projects $\ket{\psi}$ to $\ket{0}$ so that the coefficients should also include $|\braket{\psi}{0}|^2$ to ensure the norm one property of the wave function.

\subsection{Single qubit coupled to a Ohmic bath}
We consider the spin-boson model in \cite{gaspard1999non}. In the total Hamiltonian \eqref{Htot},  we choose the system operators as follows, 
\begin{equation}
    H_S = -\frac{\Delta}{2}\sigma_z, \quad S = \sigma_x,
\end{equation}
where $\Delta$ is known as the splitting energy;  $\sigma_x$ and $ \sigma_z$ are Pauli matrices. The initial state is
\begin{equation}
    \ket{\psi} = \begin{bmatrix}
        \frac{\sqrt{3}}{2}e^{-i\pi/4}\\
    \frac{1}{2}e^{i\pi/4}
    \end{bmatrix}.
\end{equation}
The spectral density of the common bath is of the exponential cutoff form
\begin{equation}\label{expcut}
    J(\omega) = \frac{\omega^3}{\omega_c^2}\exp(-\omega/\omega_c),
\end{equation}
with cut-off parameter $\omega_c$. We utilized the 7-pole approximation \red{with $\omega_c = 1$} from \cite[Table 3]{li2021markovian} so that the BCF can be accurately fitted to,
\begin{equation}
\label{eq:BCF}
    C(t) = \int_{0}^{\infty} J(\omega) e^{-i\omega t}\mathrm{d}\omega  = \sum_\mu |q_\mu|^2 e^{i\omega_\mu t}, \quad \omega_\mu \in\mathbb{C}.
\end{equation}

Since \cref{eq: M2} approximates the non-Markovian dynamics up to $O(\lambda^4 t^4)$, we can consider both strong coupling over short time intervals and weak coupling over long time intervals. For each case, we numerically evaluate the expectation with respect to Pauli observables $O_x, O_y, O_z$. 

In the strong coupling regime $\lambda^2 =0.81$ over the time interval $[0,1]$, our numerical results are summarized in \cref{fig:1spinOx,fig:1spinOy,fig:1spinOz}. The average measurements are represented by solid lines, and the shaded areas denote the standard deviation for both the noisy and NMNM trajectories. In particular, we see that the multiple-step NMNM (blue line) provides a better approximation of the ideal measurement (black line) compared with the noisy measurement(red line), even when both are generated using the same number of samples, $10^4$. 
Our error mitigation scheme used the idea of probabilistic error cancellation, and the quasi-probability causes a variance overhead compared to normal probability. Consequently, the area of the blue shadow is larger than the red shadow.

\begin{figure}[htbp]
    \centering
    \begin{subfigure}[b]{0.3\textwidth}
        \centering
        \includegraphics[width=\textwidth]{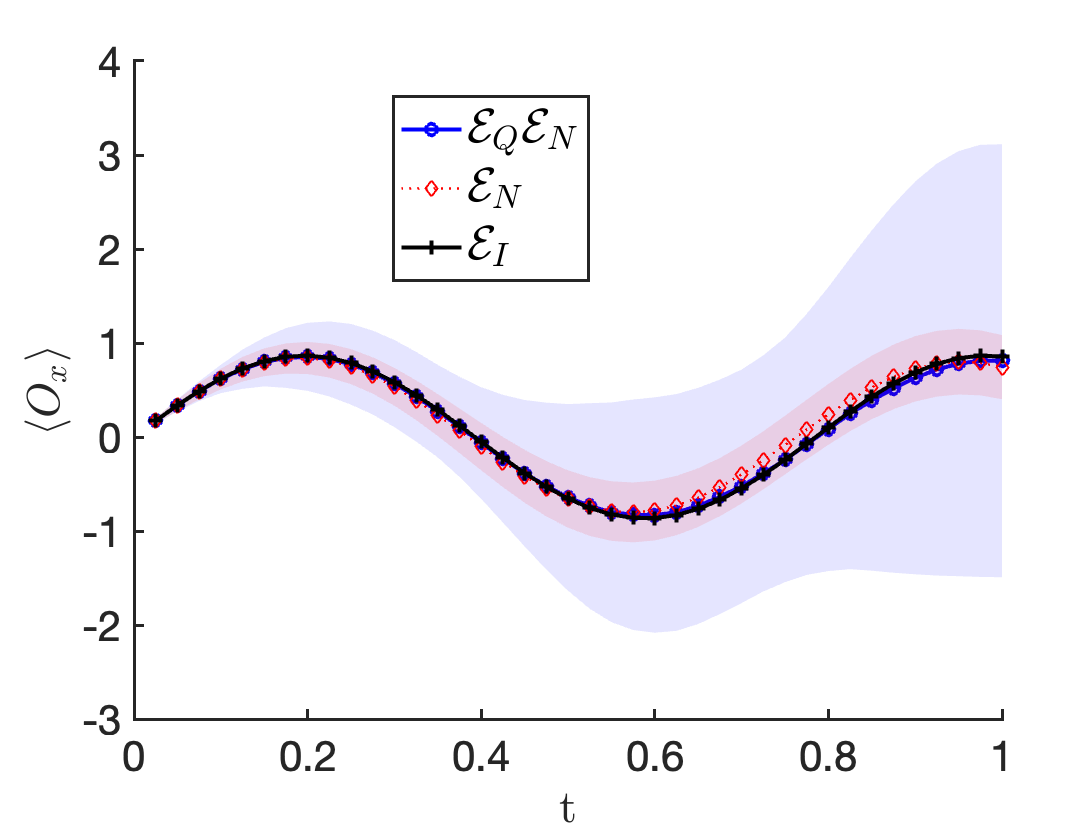}
        \caption{$O_x$}
        \label{fig:1spinOx}
    \end{subfigure}
    \hfill
    \begin{subfigure}[b]{0.3\textwidth}
        \centering
        \includegraphics[width=\textwidth]{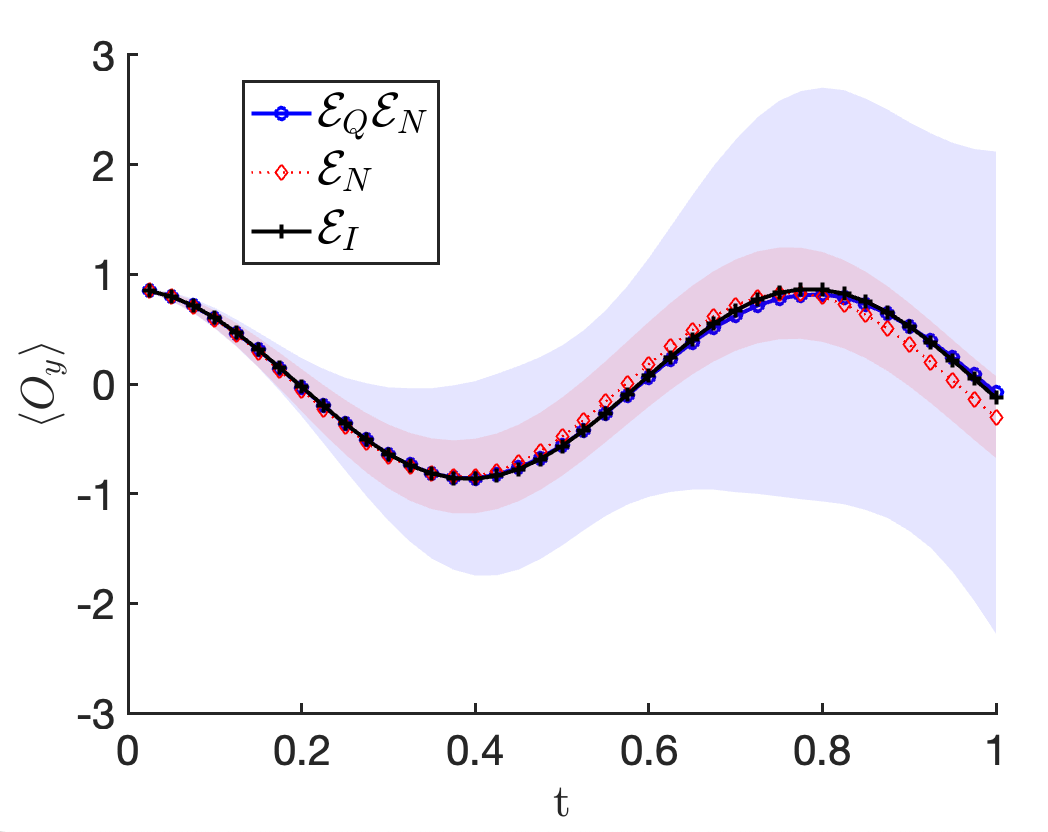}
        \caption{$O_y$}
        \label{fig:1spinOy}
    \end{subfigure}
    \hfill
    \begin{subfigure}[b]{0.3\textwidth}
        \centering
        \includegraphics[width=\textwidth]{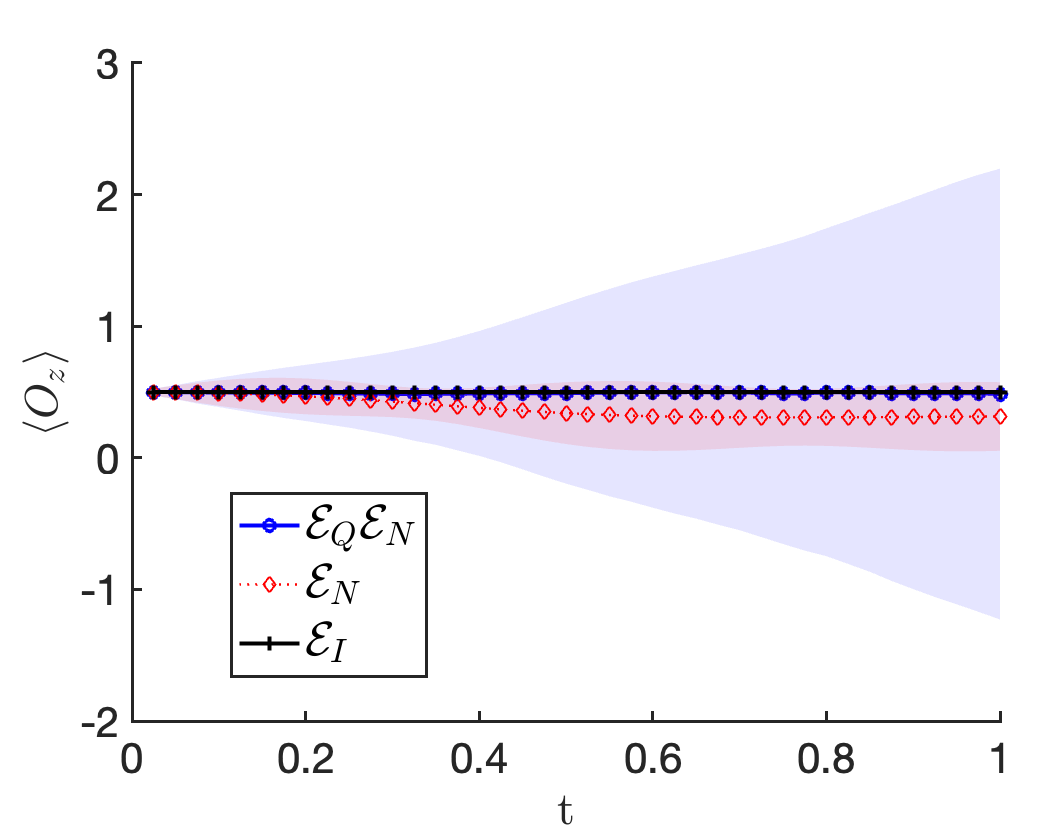}
        \caption{$O_z$}
        \label{fig:1spinOz}
    \end{subfigure}
    \caption{\textbf{Error mitigated time evolution of spin-boson model with 1 qubit.} The expectation of the density operator with respect to Pauli matrices observables $O_x, O_y, O_z$. The coupling strength $\lambda^2 = 0.81$ and the splitting energy $\Delta =8$. Monte-Carlo method is used to sample the measurement of the noisy quantum state $\rho_N(t)$, $\rho_Q(t)$ with number of samples $N_r = 10^4$. The blue and red shaded area represents the standard deviation of the population for the error mitigated trajectory $\mathcal{E}_Q\mathcal{E}_N$ and the noisy trajectories $\mathcal{E}_N$. The step size is set to $\delta t = 0.025$.}
\end{figure}

In the weak coupling regime $\lambda^2 = 0.01$, we consider a longer time interval $[0,5]$. Our results, as summarized in \cref{fig:1spinOxWeakCoupling,fig:1spinOyWeakCoupling,fig:1spinOzWeakCoupling}, demonstrate that NMNM suppresses the errors, especially effective for the Pauli-z observable $O_z$. The reduced shaded area, as compared to the strong coupling case, reflects a lower variance from the smaller coupling strength. Moreover, it indicates that although a sampling size of $10^4$ suffices to achieve robust performance in this single-qubit example for both weak and strong coupling scenarios, NMNM is more efficient in the weak coupling regime for larger systems.  
\begin{figure}[htbp]
    \centering
    \begin{subfigure}[b]{0.3\textwidth}
        \centering        \includegraphics[width=\textwidth]{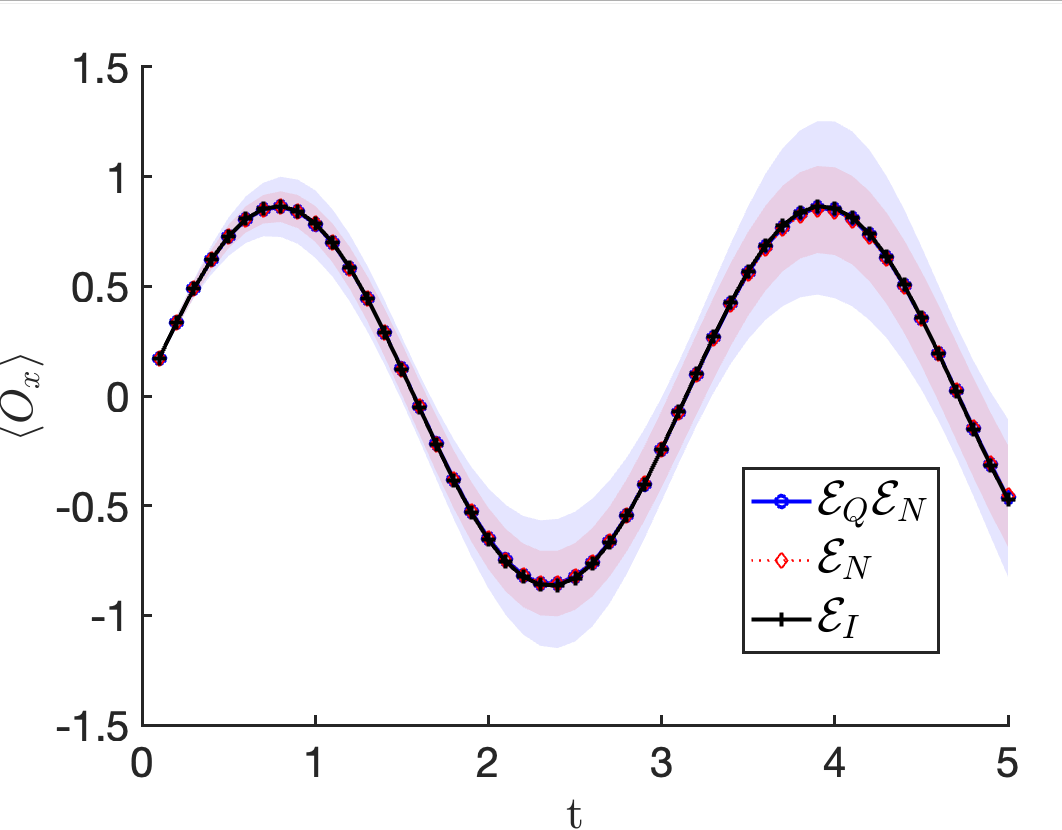}
        \caption{$O_x$}
        \label{fig:1spinOxWeakCoupling}
    \end{subfigure}
    \hfill
    \begin{subfigure}[b]{0.3\textwidth}
        \centering
        \includegraphics[width=\textwidth]{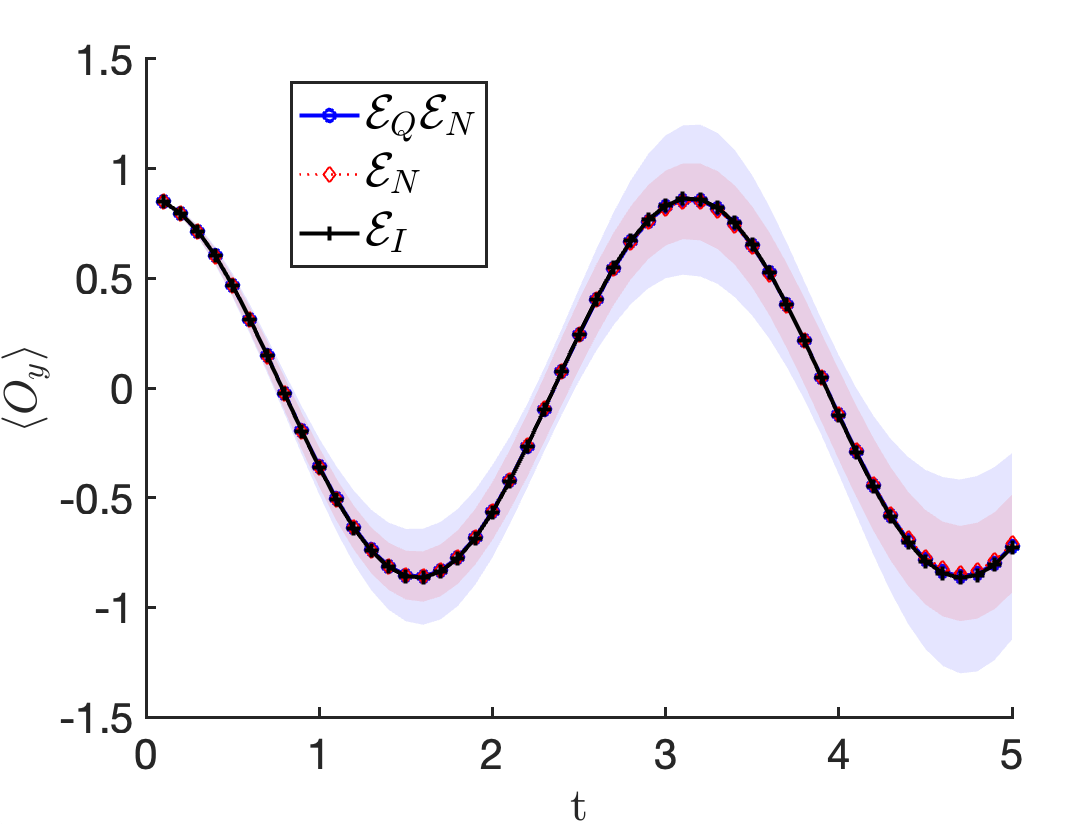}
        \caption{$O_y$}
        \label{fig:1spinOyWeakCoupling}
    \end{subfigure}
    \hfill
    \begin{subfigure}[b]{0.3\textwidth}
        \centering
        \includegraphics[width=\textwidth]{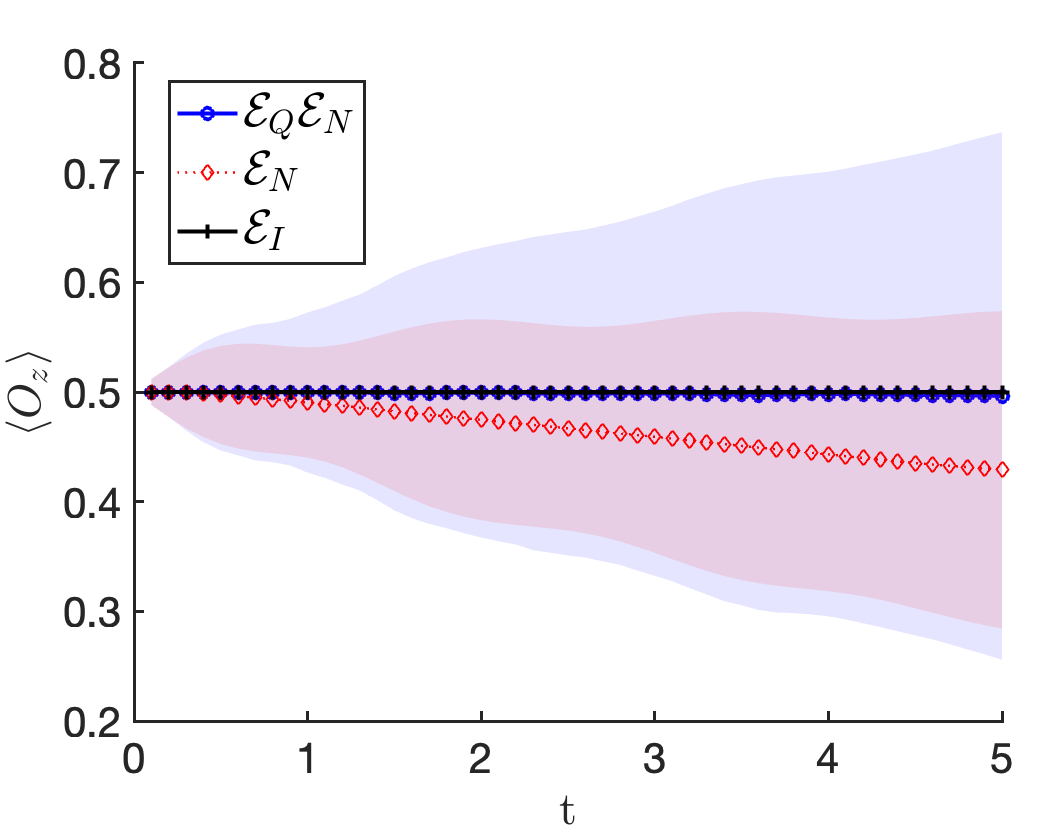}
        \caption{$O_z$}
    \label{fig:1spinOzWeakCoupling}
    \end{subfigure}
    \caption{ \textbf{Error mitigated time evolution of spin-boson model with 1 qubit.} The coupling strength $\lambda^2 = 0.01$ and the splitting energy $\Delta =2$. Monte-Carlo method is used to sample the measurement of the noisy quantum state $\rho_N(t)$ and  $\rho_Q(t)$ with $N_r = 10^4$. The blue and red shaded area represents the standard deviation of the population for the error mitigated trajectory $\mathcal{E}_Q\mathcal{E}_N$ and the noisy trajectories $\mathcal{E}_N$. The step size is set to $\delta t = 0.1$.}
    \label{fig:three_figures_weak}
\end{figure}

 To elucidate how spectral properties of the environment, characterized by $G_{\mathrm{env}}=G_{b,1}$, influence the sampling overhead, we consider an environment having BCFs in \cref{expcut} with larger coefficients $\omega_c$. We compute the resource overhead $\gamma_{\textbf{tot}}$ for various choices of these scalings and present the results in \cref{fig:Genv1}, which is in agreement with the exponential dependence revealed in \cref{thm:gamma}. For the new non-Markovian environment with larger $\teal{G_{\mathrm{env}}}$ from \cref{fig:Genv1}, we then rerun the NMNM algorithm, with results shown in \cref{fig:Genv1spinOzWeakCoupling}. As expected, the larger normalization factor leads to increased variance, compared to the results in \cref{fig:Genv1spinOzWeakCoupling}. \red{Notably, in our numerical test of~\cref{fig:Genv1,fig:Genv1spinOzWeakCoupling}, we keep $\lambda$ fixed and only modify the bath property, allowing us to attribute any observed changes to the effect of $G_{\mathrm{env}}$.  }
\begin{figure}[htbp]
    \centering
    \begin{subfigure}[b]{0.49\textwidth}
        \centering        \includegraphics[width=\textwidth]{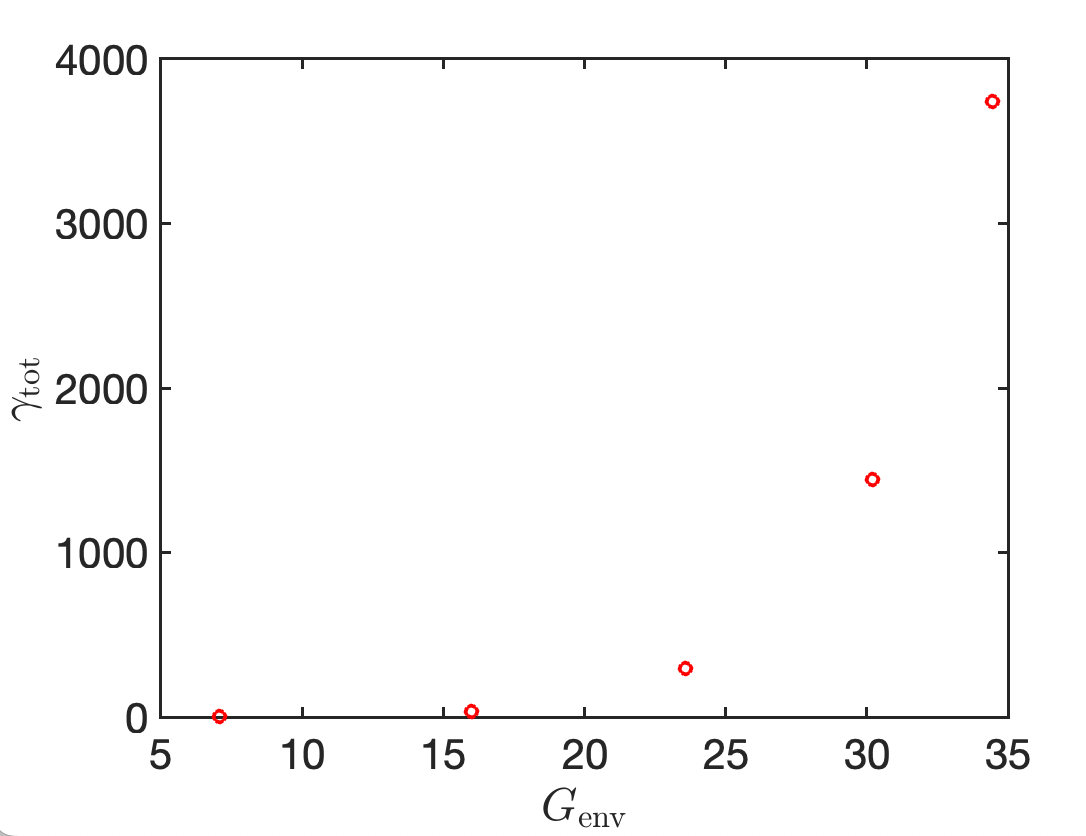}
        \caption{$\teal{G_{\mathrm{env}}}$}
        \label{fig:Genv1}
    \end{subfigure}
    \hfill
    \begin{subfigure}[b]{0.49\textwidth}
        \centering
        \includegraphics[width=\textwidth]{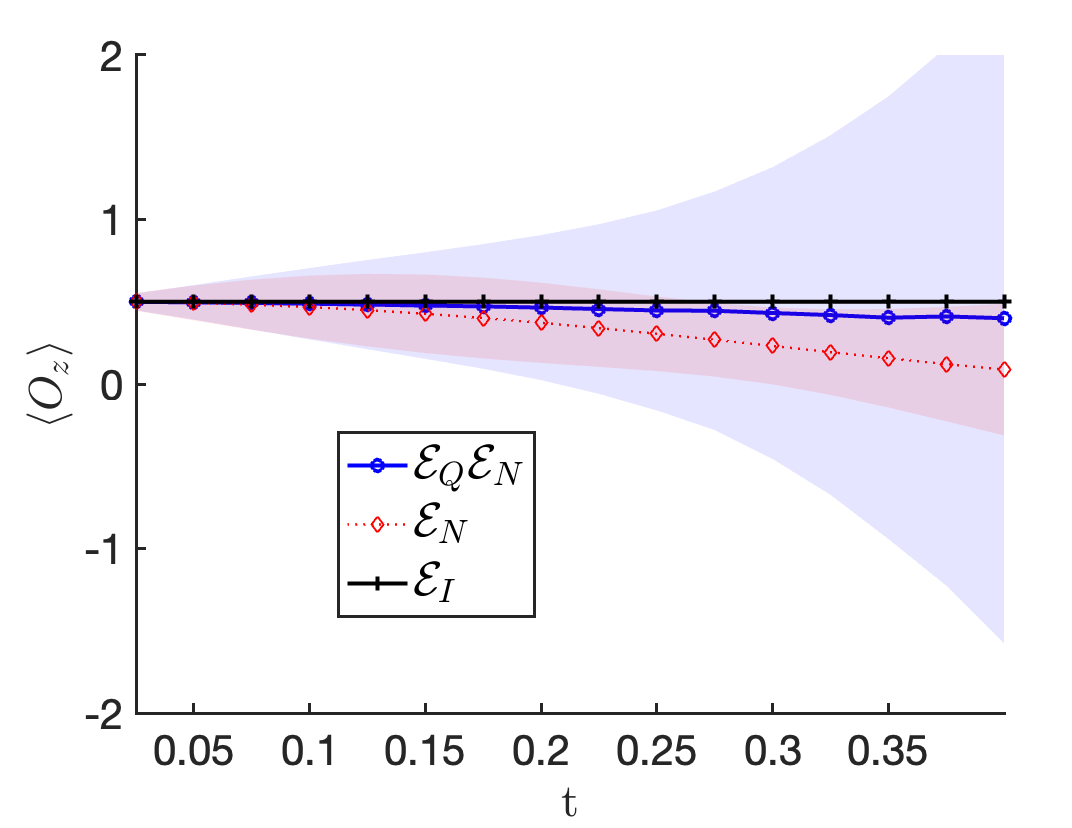}
        \caption{$O_z$}
    \label{fig:Genv1spinOzWeakCoupling}
    \end{subfigure}
    \caption{ {\textbf{Influence of $\teal{G_{\mathrm{env}}}$ on the normalization coefficient and sampling overhead.} The coupling strength $\lambda^2 = 0.81$ and the splitting energy $\Delta =8$. \cref{fig:Genv1}: The normalization factor $\gamma_{\textbf{tot}}$ is computed for environments with different exponential cutoff coefficients $\omega_c = 1, 1.5, 2, 2.5, 3$. \cref{fig:Genv1spinOzWeakCoupling}: Use the larger $\teal{G_{\mathrm{env}}}$ settings from \cref{fig:Genv1} with $\omega_c = 2$, we reapply the noise mitigation algorithm with $N_r = 10^4$. The blue and red shaded regions indicate the standard deviation of the population for the error mitigated trajectory $\mathcal{E}_Q\mathcal{E}_N$ and the noisy trajectory $\mathcal{E}_N$. The step size is set to $\delta t = 0.025$.}}
    \label{fig:Genv_gamma}
\end{figure}

\subsection{Two qubits coupled to a Ohmic bath}
Here we consider a system of two qubits that are coupled to the same bath with a spectrum with an exponential cutoff spectrum \eqref{expcut}. The total Hamiltonian  is as follows
\begin{equation}\label{eq: sb-2qubits}
    \begin{aligned}
        H_S &= \sum_{a=1,2}\frac{\Delta}{2}\sigma_z^a, \quad 
        H_{SB} = S\sum_{a=1,2;k}(g_k^a b_{a,k}+g_k^{a*}b_{a,k}^\dag),\quad S = \sum_{a = 1,2}\sigma_x^a.
    \end{aligned}
\end{equation} 
 In light of \eqref{expcut}, 
the new jump operators in \cref{That} are $\widehat{T}_\mu = \sqrt{\lambda_\mu \gamma_\mu/2}\sigma_x^\mu, \mu = 1,2$. The initial state is
\begin{equation}
    \ket{\psi} = \begin{bmatrix}
        \frac{\sqrt{3}}{2}e^{-i\pi/4}\\
    \frac{1}{2}e^{i\pi/4}
    \end{bmatrix}\otimes\begin{bmatrix}
        \frac{\sqrt{3}}{2}e^{-i\pi/4}\\
    \frac{1}{2}e^{i\pi/4}
    \end{bmatrix}.
\end{equation}

\begin{figure}[H]
    \centering
    \includegraphics[width=0.5\linewidth]{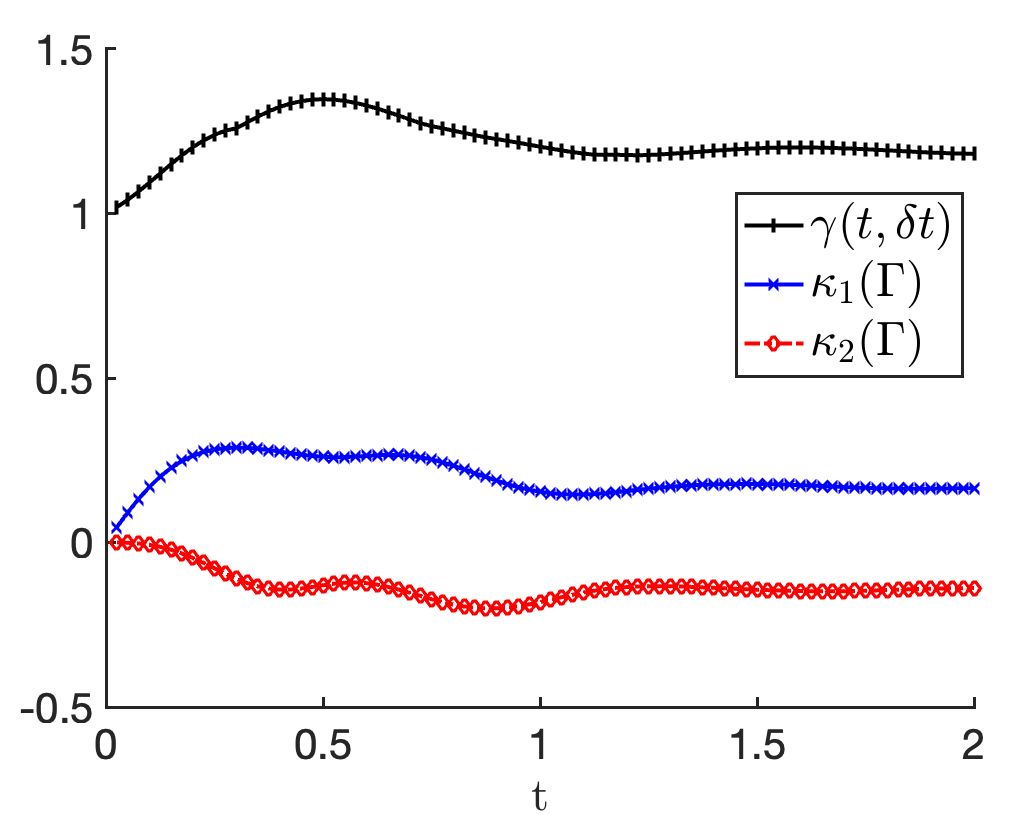}
    \caption{The eigenvalues of the coefficients $(\Gamma_{\alpha\beta})_{\alpha,\beta}$ in the incoherent noise operator $\mathcal{L}_D(t)$ in \cref{eq:LDLC}  for each time step $k\delta t,$ as well as the resource overhead $\gamma(t,\delta t) $ defined in \cref{eq:normalizationconstant}. The coupling parameter $\lambda^2 = 0.81$, the splitting energy $\Delta = 8$ and the time step $\delta t = 0.025$.}
    \label{fig:eigvals}
\end{figure}

To examine the non-Markovian nature of the noise, we computed the two eigenvalues of the coefficients $\Gamma_{\alpha\beta}$ in the incoherent noise operator $\mathcal{L}_D(t)$ in \cref{eq:LDLC}  for each time step $k\delta t$. As shown in \cref{fig:eigvals}, the second eigenvalue becomes negative over time, indicating that the corresponding dynamical map associated with \cref{eq: QME} is not divisible into completely positive maps and thus corresponds to a non-Markovian noise. Also shown in \cref{fig:eigvals} is the one-step resource overhead $\gamma(t,\delta t) $, as defined by each term in \cref{gamma_tot}.

Similar to the single qubit example, we consider both the strong and weak coupling scenarios.  With strong coupling parameter $\lambda^2 = 0.81$,  we run the noisy channel and monitor the observables 
\begin{equation}\label{eq:Oxyz}
    O_x = \frac{1}{2}\sum_{a = 1,2}\sigma_x^a, \quad O_y = \frac{1}{2}\sum_{a = 1,2}\sigma_y^a, \quad O_z = \frac{1}{2}\sum_{a = 1,2}\sigma_z^a.
\end{equation}  
We then insert $\mathcal{B}_\ell$ randomly according to the probability $p_\ell $ from \cref{eq:normalizationconstant} to implement our stochastic QEM method. The results, shown in \cref{fig:2spinOx,fig:2spinOy,fig:2spinOz},   indicate that NMNM can eliminate the non-Markovian error of the two qubits system effectively. Notably, the shaded blue area, which indicates the statistical error, grows quite quickly for large number of time steps $M$ (corresponding to multiple layers in digital quantum computing), due to the large variance. \red{Interestingly, in this regime, the error comes primarily from the statistical error due to sampling, since the averaged result is indistinguishable from the ideal circuit. In this case, the perturbation error from \cref{eq: M2} is almost negligible.  } 

\begin{figure}[!htbp]
    \centering
    \begin{subfigure}[b]{0.3\textwidth}
        \centering        \includegraphics[width=\textwidth]{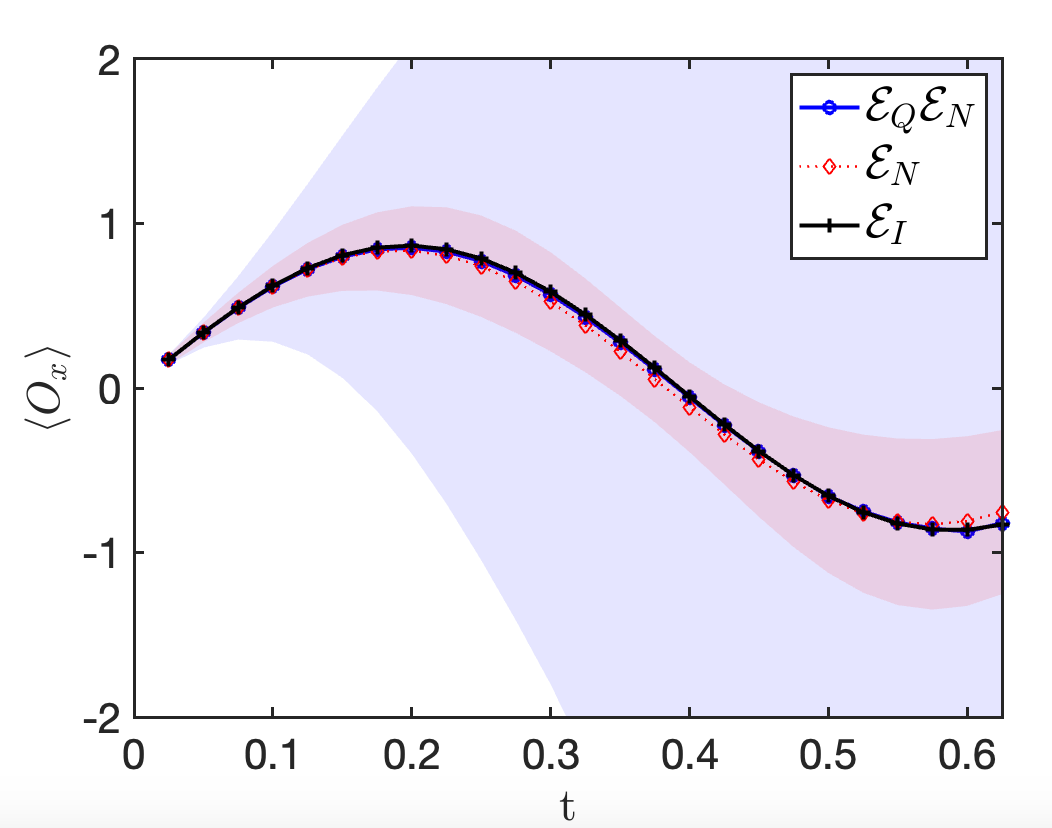}
        \caption{$O_x$}
        \label{fig:2spinOx}
    \end{subfigure}
    \hfill
    \begin{subfigure}[b]{0.3\textwidth}
        \centering
        \includegraphics[width=\textwidth]{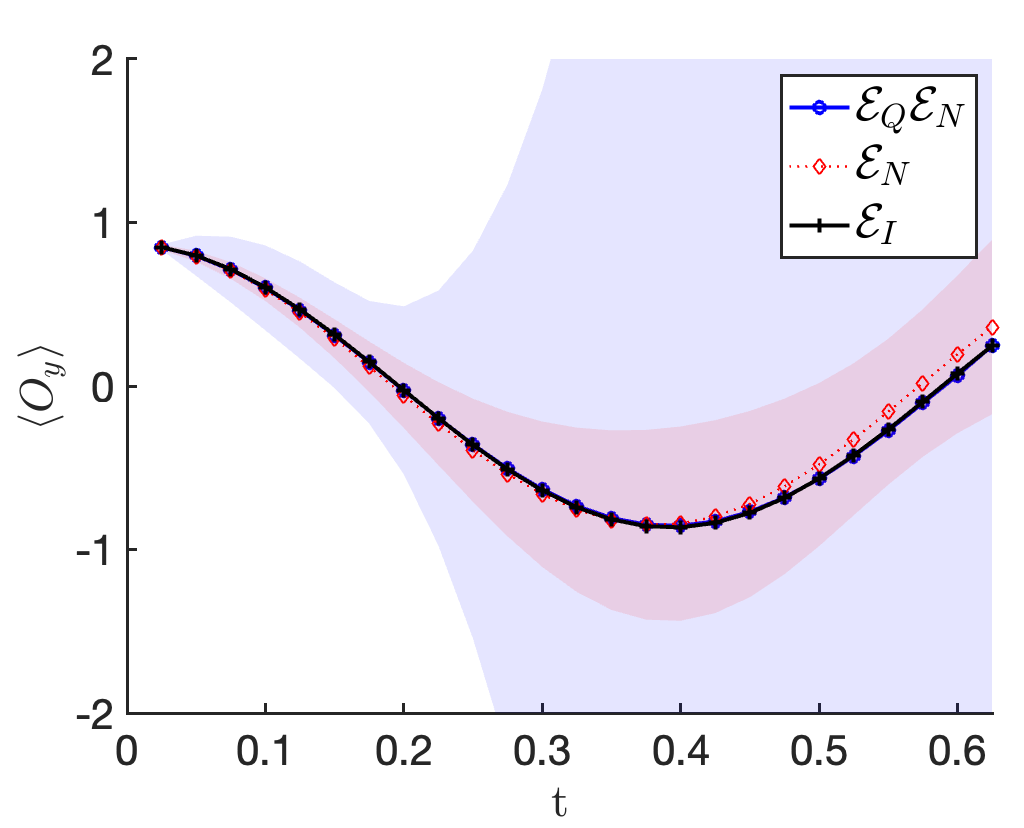}
        \caption{$O_y$}
        \label{fig:2spinOy}
    \end{subfigure}
    \hfill
    \begin{subfigure}[b]{0.3\textwidth}
        \centering
        \includegraphics[width=\textwidth]{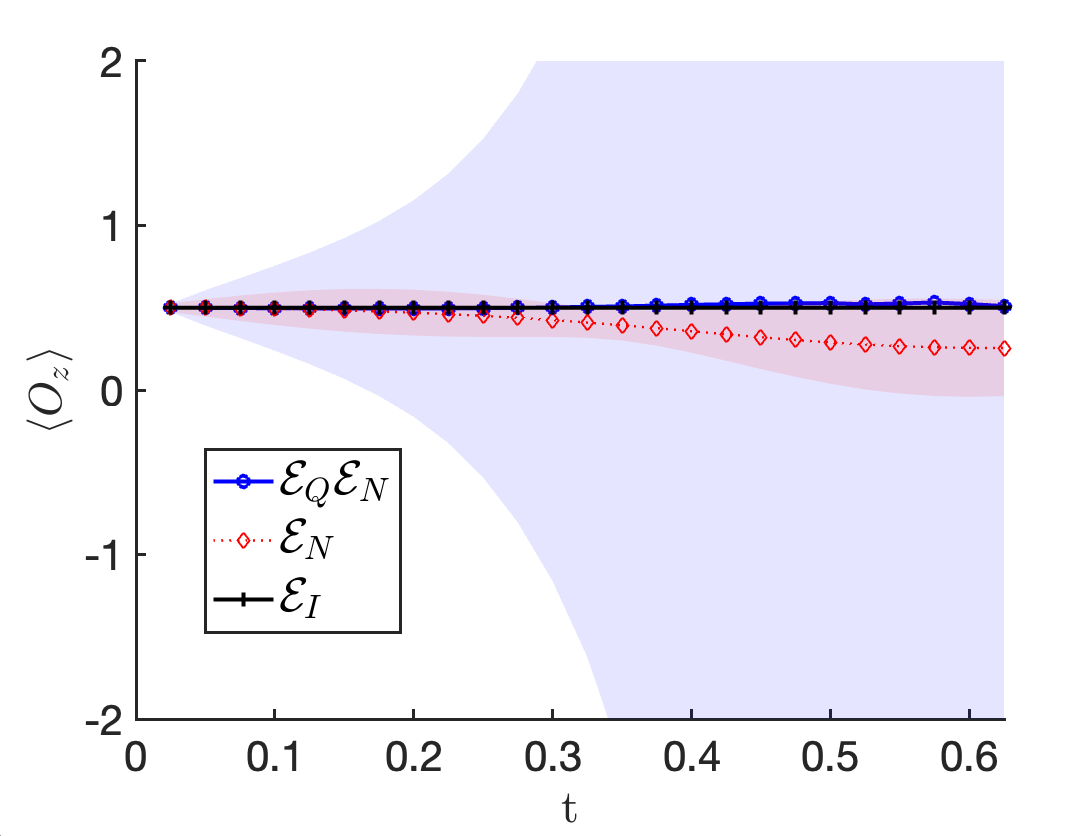}
        \caption{$O_z$}
        \label{fig:2spinOz}
    \end{subfigure}
    \caption{\textbf{Error mitigation of a 2-qubit system coupled to a boson bath \eqref{eq: sb-2qubits} with coupling strength $\lambda^2 = 0.81$.} The splitting energy  is $\Delta =8$ and the observables are defined in \cref{eq:Oxyz}. Monte-Carlo method is used to sample from the noisy quantum state $\rho_N(t)$ with sample size $N_r = 10^4$ 
    and $\rho_Q(t)$ with $N_r = 10^7$ trajectories. The blue and red shaded areas represent the standard deviation of the population for the error-mitigated trajectories $\mathcal{E}_Q\mathcal{E}_N$, the noisy trajectories $\mathcal{E}_N$, respectively. The step size in the numerical implementation is $\delta t = 0.025$. }
    \label{fig:2qubitStong}
\end{figure} 

We have observed that the circuit noise in the numerical experiment in \cref{fig:2qubitStong} induces a large statistical error, which requires a huge sampling size, $10^7$, to obtain a good estimation. This can be attributed to the strong coupling parameter, as can be seen from the estimate of the total resource overhead \eqref{gamma-bound}. To examine the impact of the coupling parameter, 
 we conducted the two-qubit experiment with a weaker coupling parameter $\lambda^2 = 0.01$. The results are shown in  \cref{fig:2spinOxWeak,fig:2spinOyWeak,fig:2spinOzWeak}. One can see that the standard deviation from the error-mitigated trajectories is much smaller and a sample size of $10^5$  already ensures a good performance of our error mitigation scheme. 
\begin{figure}[!htbp]
    \centering
    \begin{subfigure}[b]{0.3\textwidth}
        \centering
        \includegraphics[width=\textwidth]{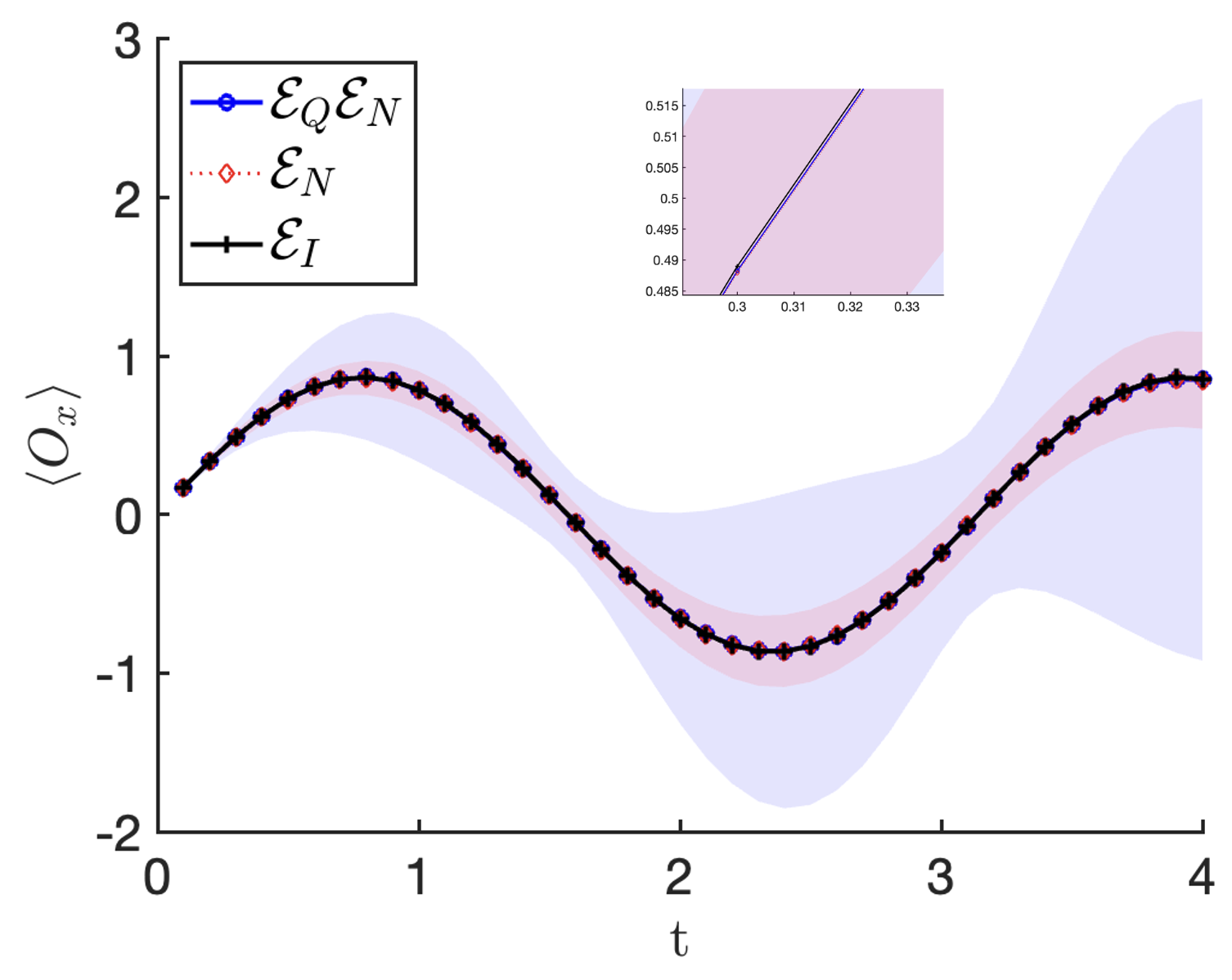}
        \caption{$O_x$}
        \label{fig:2spinOxWeak}
    \end{subfigure}
    \hfill
    \begin{subfigure}[b]{0.3\textwidth}
        \centering
        \includegraphics[width=\textwidth]{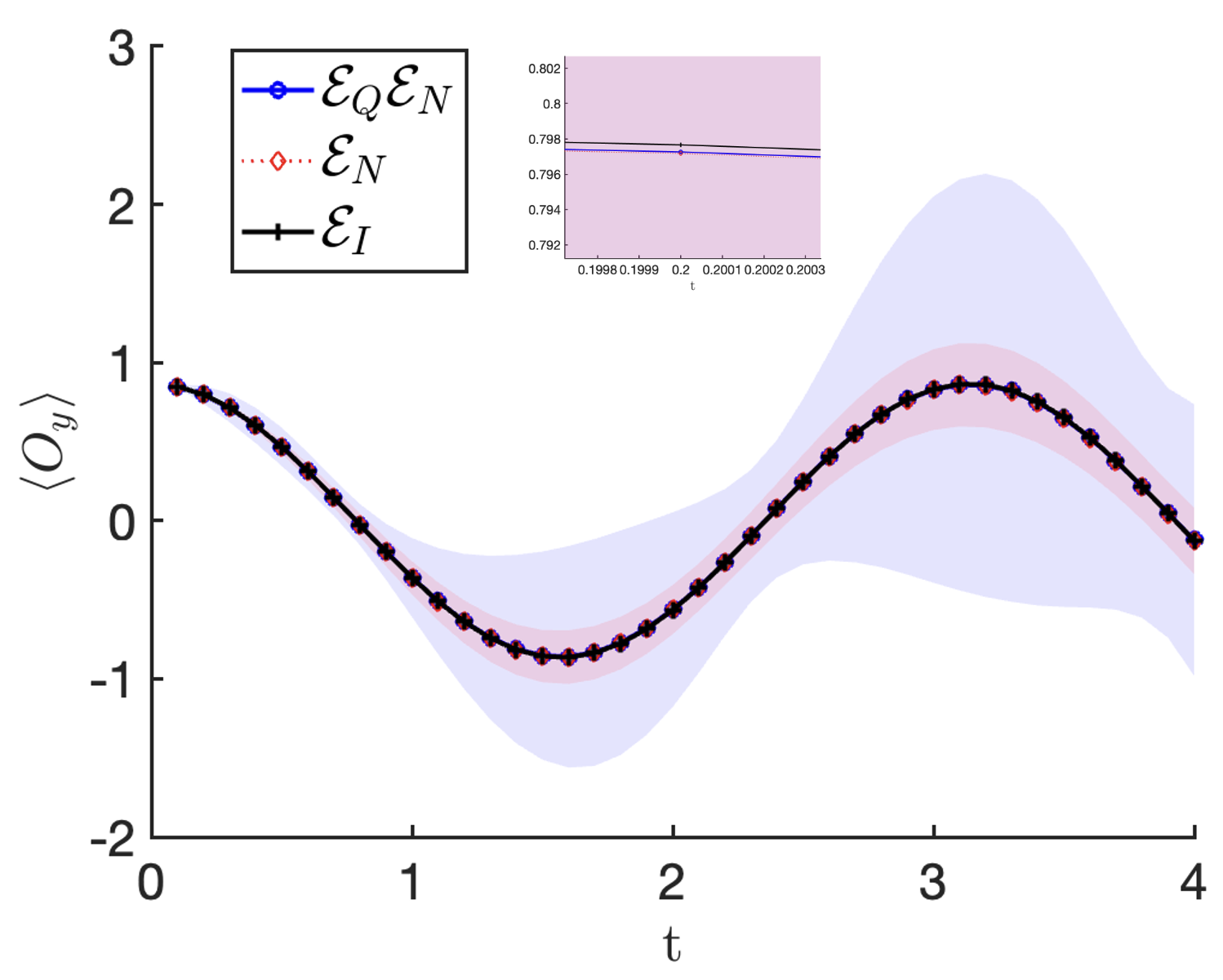}
        \caption{$O_y$}
        \label{fig:2spinOyWeak}
    \end{subfigure}
    \hfill
    \begin{subfigure}[b]{0.3\textwidth}
        \centering
        \includegraphics[width=\textwidth]{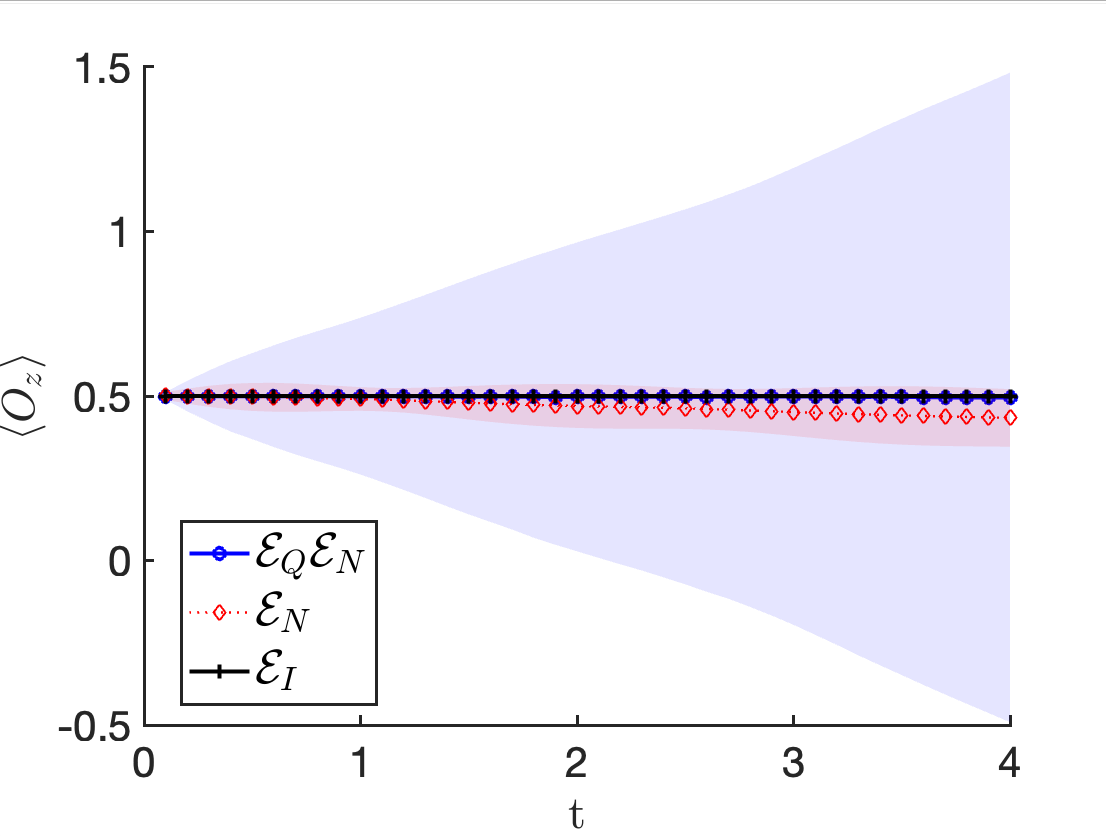}
        \caption{$O_z$}
        \label{fig:2spinOzWeak}
    \end{subfigure}
     \caption{\textbf{Error mitigation of a 2-qubit system coupled to a boson bath \eqref{eq: sb-2qubits} with coupling strength $\lambda^2 = 0.01$.} The splitting energy $\Delta =2$ in $H_S$. Monte-Carlo method is used to sample the measurement of the noisy quantum state $\rho_N(t)$ with sample size $N_r = 10^4$ and $\rho_Q(t)$ with $N_r = 10^5$. The blue and red shaded area represents the standard deviation of the population for the error mitigated trajectories $\mathcal{E}_Q\mathcal{E}_N$ and the noisy trajectories $\mathcal{E}_N$. The step size is set to $\delta t = 0.1$. }
\end{figure}

\section{Discussion}

In this work, we presented a quantum error mitigation scheme designed to suppress circuit noise on near-term quantum devices. Specifically, we extend the probability error cancellation method—originally developed for Markovian noise—to the non-Markovian regime. Our approach employs bath correlation functions as input, a common modeling framework for non-Markovian open quantum systems. In addition to detailing the algorithms and their implementations, our key contributions include establishing error bounds for both the approximation and the statistical errors, thus providing theoretical performance guarantees. Notably, the error bounds are found to be intimately tied to the spectral properties of the quantum environment.  Therefore, our results reveal how certain spectral properties of the environment critically influence the effectiveness of quantum error mitigation—an insight that could guide the selection of physical environments that ease the error mitigation efforts. 

Our mathematical formulation relies on a weak coupling assumption, justified by the fact that many current quantum computing devices operate in this regime. Furthermore, by extending the expansion in \cref{eq: M2} to include higher-order terms, we obtain a more accurate representation of the non-Markovian noise. {For example, by incorporating the operator in the $O(\lambda^4)$ terms, as presented in \cref{appd:highorder},   the error can be reduced to  $O(\lambda^6),$ thus further minizing the effective of the noise and allowing to achieve even more accurate error mitigation. }

\section{Acknowledgements }
The authors would like to acknowledge the support of this research by NSF Grants DMS-2111221 and CCF-2312456. 

\section*{Author contributions statement}
K.W. conducted the numerical experiment(s). K.W. and X.L. analyzed the methods and the numerical results.  Both authors reviewed the manuscript.

\bibliographystyle{alpha}
\bibliography{bibtex}

\appendix
\section{{Fourth-Order Representation of the non-Markovian noise}}
\label{appd:highorder}
{
The fourth-order approximation of the non-Markovian noise generator is
{\small
\begin{equation}
\begin{aligned}
&\mathcal{L}_N(t) \rho(t)= \lambda^2\sum_{\mu = 1}^{\mu_{\max}} \int_0^t\mathcal{L}(\widehat{T}_\mu , \widehat{T}_\mu(-t_2)e^{i\omega_{\mu}t_2})\mathrm{d}t_2
    \rho(t)\\
     &+ \lambda^4 \sum_{\mu_1,\mu_2=1}^{\mu_{\max}}\int_0^t\int_0^t\int_0^{t_3}\mathcal{R}(\widehat{T}_{\mu}, \widehat{T}_{\mu}(t_2-t)e^{i\omega_{\mu}(t-t_2)} ,\widehat{T}_{\mu}(t_3-t)e^{i\omega_{\mu}(t-t_3)}, \widehat{T}_{\mu}(t_4-t)e^{i\omega_{\mu}(t-t_4)})\rho(t)\mathrm{d}t_4\mathrm{d}t_3\mathrm{d}t_2 \\
     &+\lambda^4\sum_{\mu_1,\mu_2=1}^{\mu_{\max}}\int_0^t\int_0^{t_2}\int_0^{t_3}\mathcal{S}(\widehat{T}_{\mu}, \widehat{T}_{\mu}(t_2-t)e^{i\omega_{\mu}(t-t_2)}, \widehat{T}_{\mu}(t_3-t)e^{i\omega_{\mu}(t-t_3)}, \widehat{T}_{\mu}(t_4-t)e^{i\omega_{\mu}(t-t_4)})\rho(t)\mathrm{d}t_4\mathrm{d}t_3\mathrm{d}t_2 , \\
\end{aligned}
\end{equation}
where 
\begin{equation}
\begin{aligned}
    \mathcal{R}(F_\mu, G_\mu, H_\mu, J_\mu)\rho &= \{F^\dag_{\mu_1}, G_{\mu_1} \rho J^\dag_{\mu_2} H_{\mu_2}\} + \{F_{\mu_1},G_{\mu_2}\rho J_{\mu_1}^\dag H_{\mu_2}^\dag\}+  \{F_{\mu_1},G_{\mu_2}\rho J_{\mu_2}^\dag H_{\mu_1}^\dag\}+ c.c.  \\ 
    \mathcal{S}(F_\mu, G_\mu, H_\mu, J_\mu)\rho &= \{F_{\mu_1},  \rho J^\dag_{\mu_2} H_{\mu_2}G_{\mu_1}^\dag\} + \{F_{\mu_1},\rho J_{\mu_1}^\dag H_{\mu_2}^\dag G_{\mu_2}\}+  \{F_{\mu_1},\rho J_{\mu_2}^\dag H_{\mu_1}^\dag G_{\mu_2}\}+ c.c.  .\\
\end{aligned}
\end{equation}}
}

\end{document}